\documentclass[preprint,11pt]{elsarticle}
\usepackage{}

\usepackage{amsfonts}
\usepackage{amsmath}
\usepackage{amssymb}
\usepackage{latexsym}
\usepackage{graphicx}
\usepackage{ntheorem}
\usepackage[usenames]{color}

\usepackage{epsfig}

\input amssym.def
\newsymbol\rtimes 226F
\newfont{\nset}{msbm10}

\newtheorem{theo}{Theorem}[section]

\newtheorem{lemma}[theo]{Lemma}
\newtheorem{definition}[theo]{Definition}

\newtheorem{coro}[theo]{Corollary}

\theoremstyle{empty}
\theorembodyfont{\rmfamily}
\newtheorem{refproof}{theorem}

\journal{Theoretical Computer Science}

\begin{document}

\begin{frontmatter}

\title{Maximum matchings and minimum dominating sets in Apollonian networks and extended Tower of Hanoi graphs}

\author[lable1,lable2]{Yujia jin}

\author[lable1,lable3]{Huan Li}

\author[lable1,lable3]{Zhongzhi Zhang}

\ead{zhangzz@fudan.edu.cn}

\address[lable1]{Shanghai Key Laboratory of Intelligent Information
Processing, Fudan
University, Shanghai 200433, China}
\address[lable2]{School of Mathematical Sciences, Fudan University, Shanghai 200433, China}
\address[lable3]{School of Computer Science, Fudan
University, Shanghai 200433, China}
\begin{abstract}
The Apollonian networks display the remarkable power-law and small-world properties as observed in most realistic networked systems. Their dual graphs are extended Tower of Hanoi graphs, which are obtained from the Tower of Hanoi graphs by adding a special vertex linked to all its three extreme vertices. In this paper, we study analytically maximum matchings and minimum dominating sets in Apollonian networks and their dual graphs, both of which have found vast applications in various fields, e.g. structural controllability of complex networks. For both networks, we determine their matching number, domination number, the number of maximum matchings, as well as the number of minimum dominating sets.
\end{abstract}

\begin{keyword}
Maximum matching, Minimum dominating set, Apollonian network,  Tower of Hanoi graph, Matching number, Domination number, Complex network
\end{keyword}

\end{frontmatter}

\section{Introduction}

Let $\mathcal{G}=(\mathcal{V},\,\mathcal{E})$ be an unweighted connected graph with vertex set $\mathcal{V}$ of $V$ vertices and edge set $\mathcal{E}$ of $E$ edges. A matching of graph $\mathcal{G}$ is a subset $\mathcal{M}$ of $\mathcal{E}$, in which no two edges are incident to a common vertex. A vertex linked to an edge in  $\mathcal{M}$ is said to be covered by  $\mathcal{M}$, otherwise it is vacant. A maximum matching is a matching of maximum cardinality, with a perfect matching being a particular case covering all the $V$ vertices. The cardinality of a maximum matching is called matching number of graph $\mathcal{G}$. A dominating set of a graph $\mathcal{G}$ is a subset $\mathcal{D}$ of $\mathcal{V}$, such that every vertex in $\mathcal{V}\setminus \mathcal{D}$ is connected to at least one vertex in $\mathcal{D}$.  For a set $\mathcal{D}$ with the smallest cardinality, we call it a minimum dominating set (MDS), the cardinality of  which is called domination number of graph $\mathcal{G}$.

Both the maximum matching problem~\cite{LoPl86} and MDS problem~\cite{HaHeSl98} have found numerous practical applications in different areas. For example, matching number and the number of maximum matchings are relevant  in lattice statistics~\cite{Mo64}, chemistry~\cite{Vu11}, while MDS problem is closely related to routing on ad hoc wireless networks~\cite{Wu02} and multi-document summarization in sentence graphs~\cite{ShLi10}. Very recently, it was shown that maximum matchings and MDS are important analysis tools for structural controllability of complex networks~\cite{LiBa16}, based on vertex~\cite{LiSlBa11} or edge~\cite{NeVi12} dynamics. In the context of vertex dynamics, the minimum number of driving vertices necessary to control the whole network and the possible configurations for sets of driving vertices are, respectively, determined by matching number and the number of maximum matchings in a bipartite graph derived from the original network~\cite{LiSlBa11}. In the aspect of edge dynamics, the structural controllability problem can be reduced to determining domination number and the number of MDSs in the original network~\cite{NaAk12,NaAk13}, which has an advantage of needing relatively small control energy~\cite{KlShSo17}.

Due to the wide applications, it is of theoretical and practical significance to study matching number and domination number in networks, as well as the number of maximum matchings and MDSs. In the past years, concerted efforts have been devoted to developing algorithms for the problems related to maximum matchings~\cite{Er04,YaYeZh05,YaZh05,YaZh08,TeSt09,ChFrMe10,Yu13, ZhWu15, Me16, LiZh17} and minimum dominating sets~\cite{FoGrPySt08, HeIs12,dadedeMa14, GaHaK15, CoLeLi15,ShLiZh17} from  the community of mathematics and theoretical computer science. However, solving these problems is a challenge and often computationally intractable. For example, finding a MDS of a graph is NP-hard~\cite{HaHeSl98}; enumerating maximum matchings in a generic graph is arduous~\cite{Va79TCS,Va79SiamJComput}, which is \#P-complete even in a bipartite graph~\cite{Va79SiamJComput}. At present, the maximum matching and MDS problems continue to be an active research object~\cite{AsKhLi17, GoHeKr16}.

Since determining matching number and domination number, especially counting all maximum matchings and minimum dominating sets in a general graph are formidable, it is therefore of much interest to construct or find particular graphs for which the maximum matching and MDS problem can be solved exactly~\cite{LoPl86}. In this paper, we study maximum matchings and minimum dominating sets in Apollonian networks~\cite{AnHeAnDa05} and their dual graphs---extended Tower of Hanoi graphs, which are obtained from the Tower of Hanoi graphs by adding a special vertex connected to all its three extreme vertices~\cite{KlMo05}. The Apollonian networks are very useful models capturing simultaneously the remarkable scale-free~\cite{BaAl99} small-world~\cite{WaSt98} properties in realistic networks~\cite{Ne03}, and have found many applications~\cite{AnHeAnDa05}. The deterministic construction of Apollonian networks and their dual graphs allows to treat analytically their topological and combinatorial properties. For both networks, by using the decimation technique for self-similar networks~\cite{ChChY07}, we find their matching number and domination number, and determine the number of different maximum matchings and minimum dominating sets.

\section{Constructions and properties of Apollonian networks and extended Tower of Hanoi graphs}

In this section, we give a brief introduction to constructions and properties of Apollonian networks and extended Tower of Hanoi graphs---the dual of  Apollonian networks.

\subsection{Construction means and structural properties of Apollonian networks }

The Apollonian networks are generated by an iterative approach~\cite{AnHeAnDa05}.

\begin{definition}\label{Def:AP01}
Let $\mathcal{A}_n=(\mathcal{V}(\mathcal{A}_n),\mathcal{E}(\mathcal{A}_n))$, $n \geq 0$, denote the Apollonian network after $n$ iterations, with $\mathcal{V}(\mathcal{A}_n)$ and $\mathcal{E}(\mathcal{A}_n)$ being the vertex set and the edge set, respectively. Then, $\mathcal{A}_n$ is constructed as follows:

For $n=0$, $\mathcal{A}_0$ is an equilateral triangle.

For $n\geq 1$, $\mathcal{A}_n$ is obtained from $\mathcal{A}_{n-1}$ by performing the following operation: For each  existing triangle in $\mathcal{A}_{n-1}$ that was generated at the $(n-1)$th iteration, a new vertex is created and connected to the three vertices of this triangle.
\end{definition}
Fig.~\ref{FigApo01} illustrates the construction process of Apollonian networks for the first several iterations.
 \begin{figure}[htbp]
 \centering
 \includegraphics[width=0.6\textwidth]{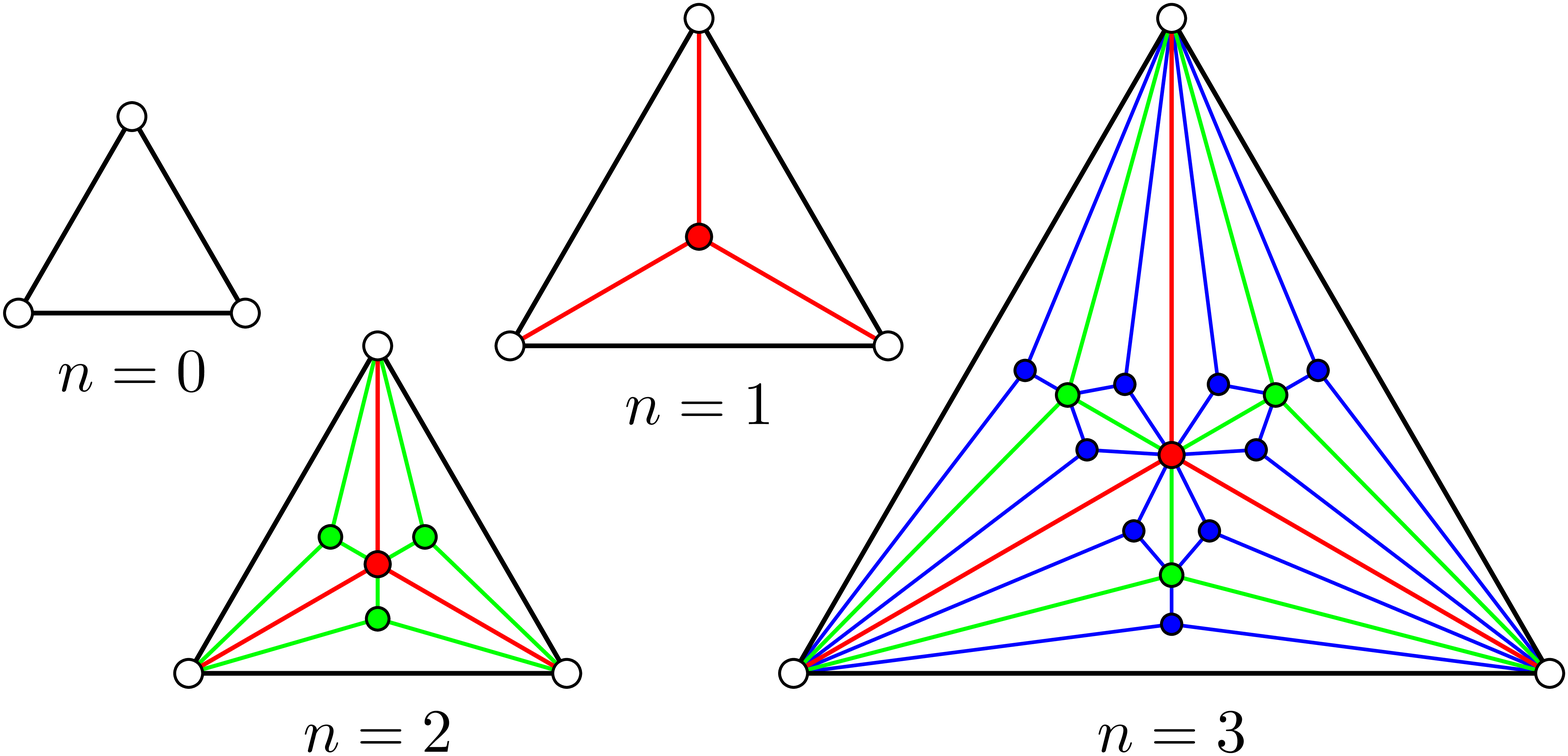}
 \caption{The iteration process for  Apollonian networks.}
 \label{FigApo01}
 \end{figure}

The  Apollonian networks are self-similar, which
suggests an alternative construction approach~\cite{ZhChZhFaGuZo08}.  For network $\mathcal{A}_{n}$, we call the three initial vertices in $\mathcal{A}_{0}$  the outmost vertices denoted by $X$, $Y$, and $Z$, respectively,  while call the vertex generated at the first iteration  the center vertex, denoted by $O$. Then, the Apollonian networks can be generated in another way as shown in Fig.~\ref{FigApo02}.

\begin{figure}
\begin{center}
\includegraphics[width=0.55\textwidth]{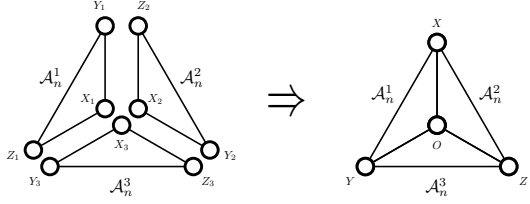}
\caption{Second construction means for  Apollonian networks. } \label{FigApo02}
\end{center}
\end{figure}

\begin{definition}\label{Def:AP02}
Given  $\mathcal{A}_{n}$, $n \geq 0$,  $\mathcal{A}_{n+1}$  is obtained from  $\mathcal{A}_{n}$  by performing the following merging operations:

(i)  Joining three copies of $\mathcal{A}_{n}$, denoted by $\mathcal{A}_{n}^{i}$, $i=1,2,3$, the three outmost vertices of which are denoted by $X_i$,  $Y_i$, and $Z_i$ ($i=1,2,3$), with $X_i$,  $Y_i$, and $Z_i$ in $\mathcal{A}_{n}^{i}$ corresponding to $X$, $Y$, and $Z$ in $\mathcal{A}_{n}$.

(ii)  Identifying $Y_1$ and $Z_2$,  $Y_3$ and $Z_1$, and $Y_2$ and $Z_3$  as the outmost vertex $X$ , $Y$,  and $Z$ of $\mathcal{A}_{n}$; and identifying  $X_1$, $X_2$ and $X_3$ as the center vertex $O$ of $\mathcal{A}_{n+1}$.
\end{definition}

Let $V_n$ and $E_n$ denote, respectively, the number of vertices and edges in $\mathcal{A}_n$.
By construction, $V_n$ and $E_n$ obey relations $V_n = 3V_{n-1}-5$ and $E_n=3E_{n-1}-3$.
With the initial conditions $V_1=4$ and $E_1=6$, we have $V_n=\frac{3^n+5}{2}$ and $E_n=\frac{3^{n+1}+3}{2}$.  Thus, $V_n$ is odd for odd $n$, and  $V_n$ is even for even $n$.

The Apollonian networks display the striking power-law and small-world characteristics of real  networks~\cite{AnHeAnDa05}. Their vertex degree obeys a power-law distribution $P(k) \sim k^{-(1+\ln 3/\ln 2)}$. 
Their mean shortest distance of all pairs of vertices increases logarithmically with the number of vertices~\cite{ZhChZhFaGuZo08}. Thus, the Apollonian networks are good models mimicking real networks, and have received considerable attention from the scientific community~\cite{ZhSuXu13,ZhWuCo14,ZhMa16}.



\subsection{Constructions and  properties of  Tower of Hanoi graphs and their extension}

In addition to the remarkable scale-free small-world properties,  other intrinsic interest of  the Apollonian networks is the relevance to the Tower of Hanoi graphs, that is, the dual graphs of Apollonian networks are extended Tower of Hanoi graphs~\cite{ZhSuXu13}.

Before introducing the dual of Apollonian networks, we first introduce the construction of the Tower of Hanoi graphs, which are derived from the Tower of Hanoi puzzle with $n$ discs~\cite{HiKlMiPeSt13}.  The puzzle consists of $n$ disks of different sizes and three pegs (towers), labeled by 0, 1, and 2, respectively. Initially, all $n$ disks are stacked on peg 0 ordered by size, with the largest on the bottom and the smallest on the top. The task is to move all disks to peg 2 complying with the following Hanoi rules:  (i) At each time step only one topmost disk may be moved; (ii) At any moment, no disk can stack on a smaller one.

A convenient and direct representation of the Tower of Hanoi problem is the graph representation. The set of $3^{n}$ configurations of the puzzle corresponds to the set of vertices. And the edges specify the legal moves of the game that obey the above rules. The resultant graphs are called the Tower of Hanoi graphs, commonly abbreviated to the Hanoi graphs. Each vertex $\xi$ in the graphs can be labeled by an $n$-tuple $\xi=(\xi_1 \xi_2 \ldots \xi_{n-1} \xi_n)$ or in the regular expression form $\xi=\xi_1 \xi_2 \ldots \xi_{n-1} \xi_n$, where $0\leq \xi_i \leq 2$ represents the peg on which the $i$th smallest disk resides. Then the labels of any pair of adjacent vertices have the Gray property, that is, they differ in exactly one position~\cite{Sa97}. We use $\mathcal{H}_n=(\mathcal{V}(\mathcal{H}_n), \mathcal{E}(\mathcal{H}_n))$  to denote  the Hanoi
graph of $n$ discs, where  $\mathcal{V}(\mathcal{H}_n)$  is the vertex set, and $\mathcal{E}(\mathcal{H}_n)$ is the edge set. Figure~\ref{FigHanoi01} shows  $\mathcal{H}_1$, $\mathcal{H}_2$ and $\mathcal{H}_3$.

\begin{figure}[htb]
\begin{center}
\includegraphics[width=0.8\textwidth]{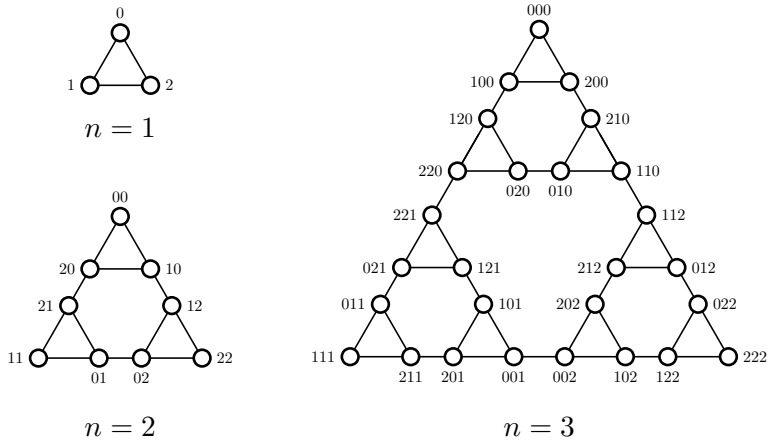}
\end{center}
\caption{Illustration of Hanoi graphs $\mathcal{H}_n$, $n=1,2,3$.}\label{FigHanoi01}
\end{figure}


In  graph  $\mathcal{H}_n$, there are $3^{n}$ vertices.  The three vertices  $00\ldots0$,   $11\ldots 1$,  and $22\ldots 2$,  are called extreme vertices, each of which corresponds to a state/configuration with all $n$ disks on one of the $3$ pegs. 
Henceforth we use   $a_n$,   $b_n$,  and $c_n$ to represent, respectively, the three extreme vertices $00\ldots0$,   $11\ldots 1$,  and $22\ldots 2$ in   $\mathcal{H}_n$.  Every extreme vertex  has $2$ neighbors, while the degree of the remaining $3^{n}-3$ vertices is 3. Thus, the number of edges in $\mathcal{H}_n$ is $\frac{3^{n+1}-3}{2}$. 

Note that $\mathcal{H}_{n+1}$, $n\ge 1$, can be obtained from three copies of $\mathcal{H}_{n}$ joined by three edges, each one connecting a pair of extreme vertices from two different replicas of $\mathcal{H}_{n}$ denoted  by
$\mathcal{H}^1_{n}$, $\mathcal{H}^2_{n}$ and $\mathcal{H}^3_{n}$, see Fig.~\ref{FigHanoi02}. The properties of the Hanoi graphs have been extensively studied~\cite{HiKlMiPeSt13,HiklSa17}.

\begin{figure}[htb]
\begin{center}
\includegraphics[width=0.5\textwidth]{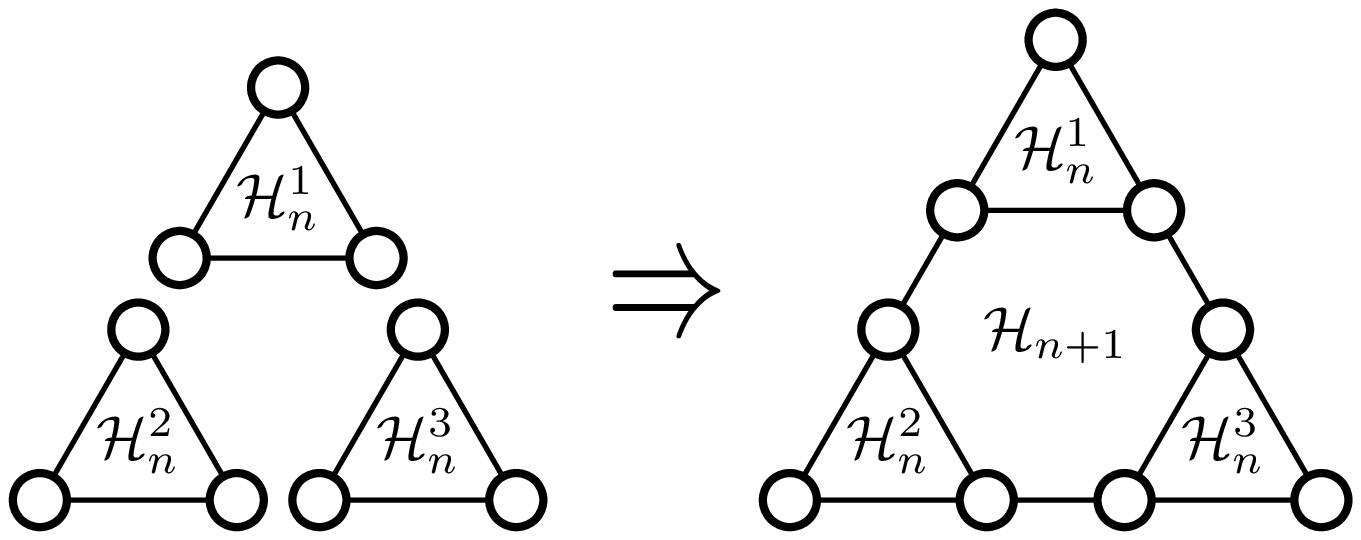}
\end{center}
\caption{Recursive construction for Hanoi graphs. }\label{FigHanoi02}
\end{figure}

The extended Tower of Hanoi graph  $\mathcal{S}_n^+=(\mathcal{V}(\mathcal{S}_n^+),\mathcal{E}(\mathcal{S}_n^+) )$, $n\geq 1$ is obtained from  $\mathcal{H}_{n}$  by adding a special vertex $s$ and three edges linking $s$ to the three extreme vertices of $\mathcal{H}_{n}$.
In graph $\mathcal{S}^+_n$,  the degree of every vertex is 3. Thus, the number of vertices and edges in  $\mathcal{S}^+_n$ is $V_n=3^n+1$ and $E_n=\frac{3^{n+1}+3}{2}$, respectively. For odd $n$,  $3^n+1$ is a multiple of 4. Fig.~\ref{FigExHanoi01} shows an  extended Tower of Hanoi graph $\mathcal{S}_{3}^+$.  It has been shown that  extended Tower of Hanoi graph  $\mathcal{S}_n^+$, $n \geq 1$, is isomorphic to the dual of  Apollonian network $\mathcal{A}_n$~\cite{ZhSuXu13}.


\begin{figure}
\centering
    \includegraphics[width=0.55\textwidth]{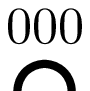}
	\caption{\label{FigExHanoi01} Illustration of an extended Tower of Hanoi graph $\mathcal{S}_{3}^+$.}
\end{figure}

\section{Maximum matchings in Apollonian networks and extended Tower of Hanoi graphs }

In this section, we study the matching number and the number of maximum matchings in Apollonian networks and extended Tower of Hanoi graphs. We will show that there exist no perfect matchings in  Apollonian networks $\mathcal{A}_n$ when $n\geq 3$, but there always exist  perfect matchings in extended Tower of Hanoi graphs $\mathcal{S}_{n}^+$ for all $n\geq 1$.   Moreover,  we will determine the number of distinct maximum matchings in both networks.

\subsection{Matching number and the number of maximum matchings in Apollonian networks }

We first study the matching number and the number of maximum matchings in Apollonian networks. 

\subsubsection{Matching number }






Although for a general graph, its matching number is not easy to determine, for Apollonian network $\mathcal{A}_n$, we can obtain it by using its self-similar structure. Let $\beta_n$ be the matching number for $\mathcal{A}_n$. In order to find $\beta_n$, we introduce some useful quantities. According to the number of covered outmost vertices, all matchings of $\mathcal{A}_n$ can be sorted into four classes: $\Omega_n^k$, $k=0,1,2,3$, with $\Omega_n^k$ representing the set of matchings covering exactly $k$ outmost vertices.  Let $\Phi_n^k$, $k=0,1,2,3$ be subset of $\Omega_n^k$ consisting of all elements with  maximum size (cardinality).  Let $\beta^k_n$, $k=0,1,2,3$, denote the size of each matching in   $\Phi_n^k$.  

\begin{lemma}\label{lem01}
The matching number of network $\mathcal{A}_n$, $n\geq1$, satisfies\par
$\beta_n = \max \{\beta^0_n,\beta^1_n,\beta^2_n,\beta^3_n\}$.
\end{lemma}
After reducing the problem of determining $\beta_n$ to computing $\beta^0_n,\beta^1_n,\beta^2_n$ and $\beta^3_n$, we next evaluate their values.  

\begin{lemma}\label{lem02}
For two successive Apollonian networks $\mathcal{A}_n$ and   $\mathcal{A}_{n+1}$,  $n\geq1$,  the following relations hold.
\begin{equation}\label{M1}
\beta^0_{n+1}=\max\{3\beta^0_n,2\beta^0_n+\beta^1_n\},
\end{equation}
\begin{equation}\label{M2}
\beta^1_{n+1} = \max\{2\beta^0_n+\beta^1_n,2\beta^0_n+\beta^2_n,\beta^0_n +2\beta^1_n\},
\end{equation}
\begin{equation}\label{M3}
\beta^2_{n+1}= \max\{2\beta^0_n+\beta^3_n,\beta^0_n+\beta^1_n+\beta^2_n,3\beta^1_n,\beta^0_n +2\beta^1_n,2\beta^0_n+\beta^2_n\},
\end{equation}
\begin{equation}\label{M4}
\beta^3_{n+1} = \max\{\beta^0_n+\beta^1_n+\beta^2_n,\beta^0_n+\beta^1_n+\beta^3_n,\beta^0_n+2\beta^2_n,3\beta^1_n,2\beta^1_n+\beta^2_n\}.
\end{equation}
\end{lemma}
\begin{proof}
This lemma  can be  proved graphically.    
Note that  each set $\Omega_{n+1}^k$, $k=0,1,2,3$, can be further classified into two groups of matchings: one group includes those matchings covering the center vertex in $\mathcal{A}_{n+1}$, while the other group contains those matchings with the center vertex vacant.  According to the classification, Figs.~\ref{fig:01},~\ref{fig:02},~\ref{fig:03} and~\ref{fig:04}  show, respectively, configurations of   matchings
for $\mathcal{A}_{n+1}$ belonging  to $\Omega_{n+1}^k$, $k=0,1,2,3$, that comprise all the  matchings in $\Phi_{n+1}^k$.   In Figs.~\ref{fig:01},~\ref{fig:02},~\ref{fig:03} and~\ref{fig:04}, only the outmost and center vertices are shown explicitly, with filled circles being covered vertices and empty circles  being vacant vertices.  From Figs.~\ref{fig:01},~\ref{fig:02},~\ref{fig:03} and~\ref{fig:04}, we can establish the relations in
 Eqs.~\eqref{M1},~\eqref{M2},~\eqref{M3} and~\eqref{M4}.
\end{proof}

\begin{figure}
\centering
 \includegraphics[width=0.6\textwidth]{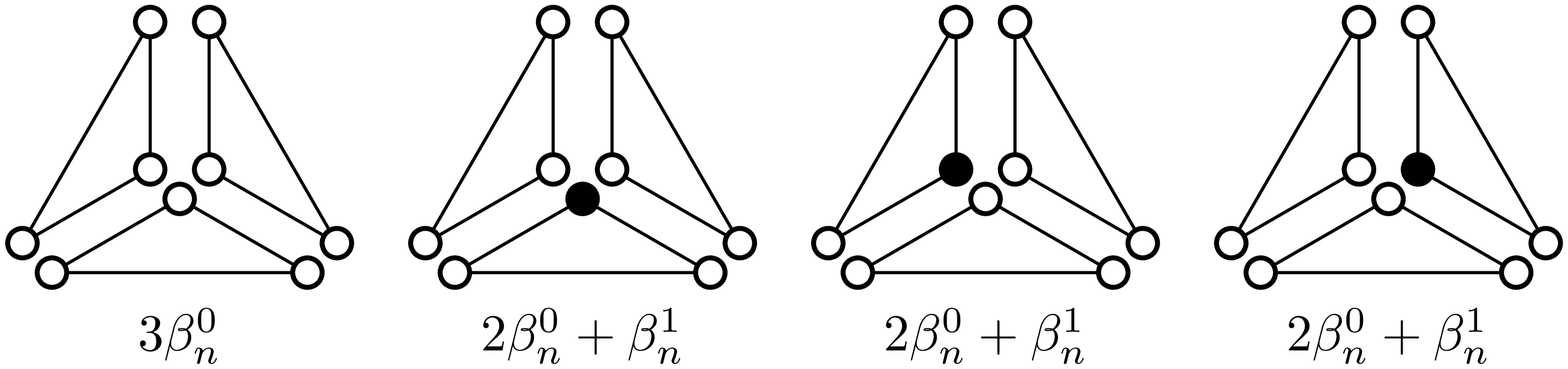}
\caption{\label{fig:01}  Configurations of matchings for $\mathcal{A}_{n+1}$   belonging to  $\Omega_{n+1}^0$, which contain all matchings  in $\Phi_{n+1}^0$.}
\end{figure}
\begin{figure}
\centering
\includegraphics[width=\textwidth]{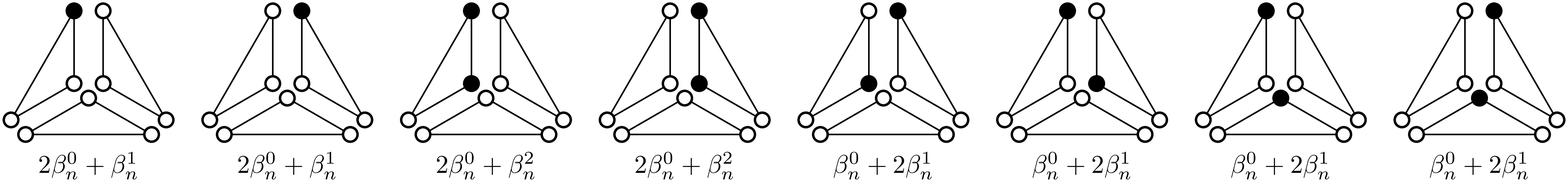}
\caption{\label{fig:02} Configurations of matchings for $\mathcal{A}_{n+1}$   belonging to  $\Omega_{n+1}^1$, which contain all matchings in $\Phi_{n+1}^1$.  In view of the symmetry, we only illustrate those matchings of $\mathcal{A}_{n+1}$, for each of which $X$ is covered, while $Y$ and $Z$  are vacant, but omit  matchings that cover only $Y$ (or $Z$), but not cover  $X$ and $Z$ (or $X$ and $Y$). }
\end{figure}
\begin{figure}
\centering
\includegraphics[width=\textwidth]{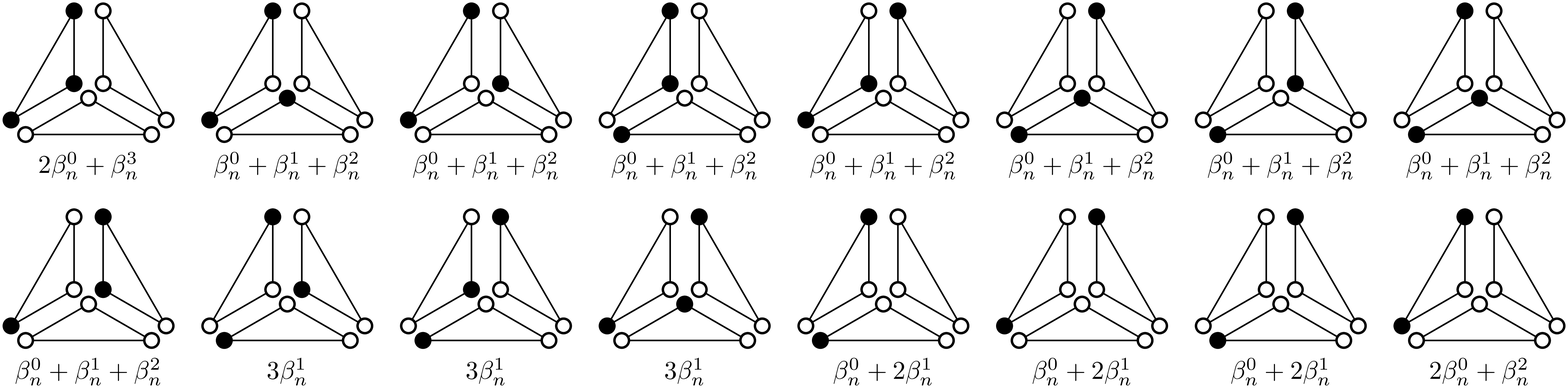}
\caption{\label{fig:03} Configurations of matchings for $\mathcal{A}_{n+1}$   belonging to  $\Omega_{n+1}^2$, which contain all matchings in $\Phi_{n+1}^2$. By symmetry, as in Fig.~\ref{fig:02}, we only illustrate those matchings of $\mathcal{A}_{n+1}$, for each of which $X$ and $Y$ are covered, while  $Z$  is vacant.}
\end{figure}
\begin{figure}
\centering
    \includegraphics[width=\textwidth]{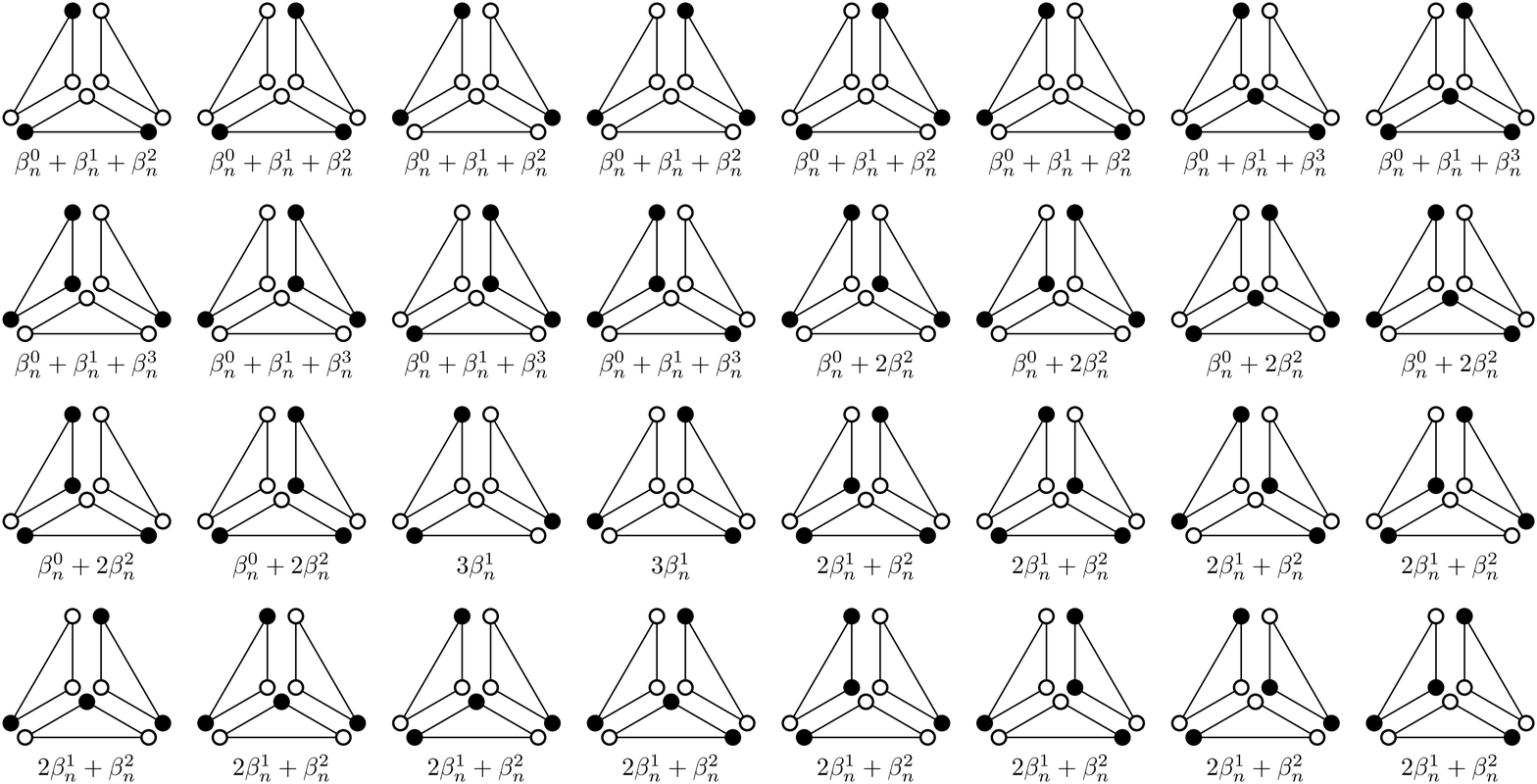}
	\caption{\label{fig:04} Configurations of matchings for $\mathcal{A}_{n+1}$   belonging to  $\Omega_{n+1}^3$, which contain all matchings in $\Phi_{n+1}^3$.}
\end{figure}

For small $n$,   $\beta^0_n$, $\beta^1_n$ ,  $\beta^2_n$ and $\beta^3_n$ can by directly computed by hand.  For example, for
$n=1,2,3$,  we have $(\beta^0_1,\beta^1_1,\beta^2_1,\beta^3_1)=(0,1,1,2)$, $(\beta^0_2,\beta^1_2,\beta^2_2,\beta^3_2)=(1,2,3,3)$, and $(\beta^0_3,\beta^1_3,\beta^2_3,\beta^3_3)=(4,5,6,7)$, respectively. For $n\geq 3$, the following lemma establishes  the relations between $\beta^0_n$, $\beta^1_n$ ,  $\beta^2_n$ and $\beta^3_n$.


\begin{lemma}\label{lem04}
For Apollonian network $\mathcal{A}_n$, $n \geq 3$, quantities $\beta^0_n,\beta^1_n,\beta^2_n,\beta^3_n$, satisfy the following relations:
\begin{equation}\label{M5}
\beta^0_n+3=\beta^1_n+2=\beta^2_n+1=\beta^3_n.
\end{equation}
\end{lemma}
\begin{proof}
By induction. When $n=3$, the result is true.

Let us suppose that the relations hold true for $n=k$, i.e., $\beta^0_k+3=\beta^1_k+2=\beta^2_k+1=\beta^k_n$. For $n=k+1$, by Lemma~\ref{lem02} we get $\beta^0_{k+1}+3=\beta^1_{k+1}+2=\beta^2_{k+1}+1=\beta^3_{k+1}$. This completes the proof.
\end{proof}

\begin{theo}\label{Thm01}
For Apollonian network $\mathcal{A}_n$, $n \geq 3$,  its matching number is
\begin{equation}\label{M6}
\beta_n=\frac{3^{n-1}+5}{2}.
\end{equation}
\end{theo}
\begin{proof}
Lemma~\ref{lem04} combined with Lemma~\ref{lem01} shows that $\beta_n=\beta^3_n$.  From Lemma~\ref{lem02}, we obtain the recursive relation governing $\beta^3_{n+1}$ and   $\beta^3_n$ as  $\beta^3_{n+1} = \max\{\beta^0_n+\beta^1_n+\beta^2_n,\beta^0_n+\beta^1_n+\beta^3_n,\beta^0_n+2\beta^2_n,3\beta^1_n,2\beta^1_n+\beta^2_n\}=3(\beta^3_n-3)+4$, which, together with  the initial value $\beta^3_3=7$, is solved to yield $\beta_n=\beta^3_n=\frac{5+3^{n-1}}{2}$.   
\end{proof}

The matching number $\beta_n=\frac{3^{n-1}+5}{2}$ for Apollonian network $\mathcal{A}_n$,  $n\geq3$, is much smaller than half the number of vertices $\frac{V_n}{2}=\frac{3^{n}+5}{4}$,  indicating that no perfect matching exists in $\mathcal{A}_n$.

\begin{coro}
For Apollonian network $\mathcal{A}_n$, $n \geq 3$, the maximum size of a matching  covering exactly $0$,  $1$, and $2$ outmost vertices, is
\begin{equation}\label{M7}
\beta^0_n = \frac{3^{n-1}-1}{2},
\end{equation}
\begin{equation}\label{M8}
\beta^1_n = \frac{3^{n-1}+1}{2}
\end{equation}
and
\begin{equation}\label{M9}
\beta^2_n = \frac{3^{n-1}+3}{2},
\end{equation}
respectively.
\end{coro}
\begin{proof}
By Eqs.~\eqref{M5} and~\eqref{M6}, the result is concluded immediately.
\end{proof}


\subsubsection{Number of maximum matchings}

Let $\tau_n$ be the number of maximum matchings of $\mathcal{A}_n$.  To determine $\tau_n$, we introduce three more quantities.  Let $\varphi_n$ be the number of matchings of $\mathcal{A}_n$ that are maximum among those matchings for which  all the three outmost vertices $X$, $Y$, and $Z$ are vacant.  Let $\theta_n$ be the number of matchings of $\mathcal{A}_n$ that are maximum among those matchings for which  $X$ is covered, while $Y$ and $Z$  are vacant. By symmetry,   $\theta_n$ equals the number of matchings of $\mathcal{A}_n$ that are maximum among those matchings for which $Y$ ($Z$) is covered, while $X$ and $Z$   ($X$ and $Y$) are vacant.  Let $\phi_n$ be the number of matchings of $\mathcal{A}_n$ that are maximum among those matchings for which $X$ is vacant, while $Y$ and $Z$  are covered.  Obviously,  $\phi_n$ is equal to the number of matchings of $\mathcal{A}_n$ that are maximum among those matchings for which $Y$  ($Z$) is vacant, while $X$ and $Z$  ($X$ and $Y$)  are covered. For small $n$, quantities $\varphi_n$, $\theta_n$, $\phi_n$, and $\tau_n$ can be easily determined by using a computer.  For example,  in the case of  $n=3$, $\varphi_3=108$, $\theta_3=246$, $\phi_3=480$, and $\tau_3=738$.  For large $n$, they can be determined recursively as follows.



\begin{theo}\label{Thm02}
For Apollonian network $\mathcal{A}_n$, $n \geq3$, the four quantities $\varphi_n$, $\theta_n$, $\phi_n$, and $\tau_n$ can be recursively determined according to  the following relations:
\begin{equation}\label{M10}
\varphi_{n+1} = 3\theta_n\varphi_n^2,
\end{equation}
\begin{equation}\label{M11}
\theta_{n+1} = 2 \phi_n\varphi_n^2+4\theta_n^2\varphi_n,
\end{equation}
\begin{equation}\label{M12}
\phi_{n+1} = \tau_n\varphi_n^2+8\phi_n\theta_n\varphi_n+3\theta_n^3,
\end{equation}
\begin{equation}\label{M13}
\tau_{n+1}=6\tau_n\theta_n\varphi_n+6\varphi_n\phi_n^2+12\theta_n^2\phi_n,
\end{equation}
with initial values $\varphi_3=108$, $\theta_3=246$, $\phi_3=480$, and $\tau_3=738$.
\end{theo}
\begin{proof}
First we show that all matchings listed in Figs.~\ref{fig:01}, ~\ref{fig:02}, ~\ref{fig:03}, and ~\ref{fig:04} are truly disjoint for $n\ge3$. For this purpose, it is sufficient to show that when $n\ge3$, for any matching $\mathcal{M}$ of $\mathcal{A}_{n+1}$ belonging to $\Phi_{n+1}^k$, $k=1,2,3$, no edge $e$ connecting the outmost vertex $X_i$ and another outmost vertex $Y_i$  (or $Z_i$ ) in  $\mathcal{A}^i_{n}$ ($i=1,2,3$) is in  $\mathcal{M}$.  This can be proved as follows.

We only prove the case that such an edge $e$ is inexistent in any matching  $\mathcal{M}$  in $\Phi_{n+1}^1$. Otherwise,  suppose that the edge $e$ linked to $X_i$ and $Y_i$  is in $\mathcal{M}$. We distinguish two cases: (i) $Z_i$ is vacant; (ii) $Z_i$ is covered by an edge in $\mathcal{A}^i_{n}$. For  case (i),  all  edges belonging to  $\mathcal{A}^i_{n}$ and $\mathcal{M}$ at the same time constitute a matching of $\mathcal{A}^i_{n}$ that is in $\Phi_{n}^2$. If we delete $e$, then all the remaining edges simultaneously belonging to  $\mathcal{A}^i_{n}$ and $\mathcal{M}$ form a matching of $\mathcal{A}^i_{n}$ that is in $\Omega_{n}^0$. This means $\beta^0_n+1\ge\beta^2_n$,  in contradiction with the fact $\beta^0_n+2=\beta^2_n$ obtained in Lemma~\ref{lem04}. Analogously, we can prove  case (ii).

In a similar argument, we can prove that  such an edge $e$ does not exist in any matching  $\mathcal{M}$  in $\Phi_{n+1}^2$ or $\Phi_{n+1}^3$.




Then, to prove Eqs.~\eqref{M10},~\eqref{M11},~\eqref{M12}, and~\eqref{M13}, it suffices to count all cases that yield maximum matchings in $\Phi_{n}^k$, $k=0,1,2,3$, listed in Figs.~\ref{fig:01},~\ref{fig:02},~\ref{fig:03} and~\ref{fig:04},  respectively.

We here only prove  Eq.~\eqref{M12}, since the other three  equations can be proved in a similar way.   

According to Eq.~\eqref{M3} and Lemma~\ref{lem04},  all the  possible configurations of maximum matchings with size $\beta^2_{n+1}$ are those having size $2\beta^0_n+\beta^3_n$, $\beta^0_n+\beta^1_n+\beta^2_n$ or $3\beta^1_n$. From Fig.~\ref{fig:03}, we can see that among all maximum matchings of  $\mathcal{A}_n$ covering $X$ and $Y$ but not  $Z$,  the number of those with size $2\beta^0_n+\beta^3_n$, $\beta^0_n+\beta^1_n+\beta^2_n$, $3\beta^1_n$ is equal to $\tau_n\varphi_n^2$, $8\phi_n\theta_n\varphi_n$, and $3\theta_n^3$, respectively. By definition of $\phi_{n+1}$,  Eq.~\eqref{M12} holds.
\end{proof}

Theorem~\ref{Thm02} shows that the number of maximum matchings of Apollonian network $\mathcal{A}_n$ can be calculated in $\mathcal{O}(\ln V_n)$ time.

Applying Eqs.~\eqref{M10}-\eqref{M13}, the values of $\varphi_n$, $\theta_n$, $\phi_n$,  and $\tau_n$ can be determined for small $n$ as listed in Table~\ref{MatchingNo}. 
However, it is very difficult to obtain closed-form expressions for these quantities.

\begin{table*}
\caption{The first several values of $\varphi_n$, $\theta_n$, $\phi_n$,  and $\tau_n$.}\label{MatchingNo}
\normalsize
\centering
\begin{small}
\begin{tabular}{|c|c|c|c|c|c|}
\hline
\raisebox{-0.5ex}{$n$} & \raisebox{-0.5ex}{1} & \raisebox{-0.5ex}{2} & \raisebox{-0.5ex}{3} & \raisebox{-0.5ex}{4} & \raisebox{-0.5ex}{5} \\[0.5ex]
\hline
\hline
\raisebox{-0.5ex}{$V_n$} & \raisebox{-0.5ex}{4} & \raisebox{-0.5ex}{7} & \raisebox{-0.5ex}{16} & \raisebox{-0.5ex}{43} & \raisebox{-0.5ex}{124} \\[0.5ex]
\hline
\raisebox{-0.5ex}{$\varphi_n$} & \raisebox{-0.5ex}{1} & \raisebox{-0.5ex}{3} & \raisebox{-0.5ex}{108} & \raisebox{-0.5ex}{8,608,032} & \raisebox{-0.5ex}{8,300,560,282,271,896,633,344}\\[0.5ex]
\hline
\raisebox{-0.5ex}{$\theta_n$} & \raisebox{-0.5ex}{1} & \raisebox{-0.5ex}{4} & \raisebox{-0.5ex}{246} & \raisebox{-0.5ex}{37,340,352} & \raisebox{-0.5ex}{71,022,198,720,317,181,345,792} \\[0.5ex]
\hline
\raisebox{-0.5ex}{$\phi_n$} & \raisebox{-0.5ex}{1} & \raisebox{-0.5ex}{3} & \raisebox{-0.5ex}{480} & \raisebox{-0.5ex}{155,289,960} & \raisebox{-0.5ex}{601,114,712,194,856,725,217,280} \\[0.5ex]
\hline
\raisebox{-0.5ex}{$\tau_n$} & \raisebox{-0.5ex}{3} & \raisebox{-0.5ex}{23} & \raisebox{-0.5ex}{738} & \raisebox{-0.5ex}{615,514,464} & \raisebox{-0.5ex}{5,030,805,301,520,123,200,352,256} \\[0.5ex]
\hline
\end{tabular}
\end{small}
\end{table*}


\subsubsection{Asymptotic growth constant for the number of maximum matchings}

Table~\ref{MatchingNo}  shows that number of maximum matchings  $\tau_n$ grows exponentially with $V_n$. Let $Z(\mathcal{A}_n)=   \frac{\ln \tau_n}{V_{n}}$,  then we can define a constant describing this exponential growth:
\begin{equation}\label{MMcons}
Z(\mathcal{A})= \lim_{{n}\rightarrow\infty}Z(\mathcal{A}_n)=  \lim_{{n}\rightarrow\infty} \frac{\ln \tau_n}{V_{n}}.
\end{equation}
Below we will show that this limit exists. To evaluate the asymptotic growth constant in Eq.~\eqref{MMcons}, we need the following lemma.

\begin{lemma}\label{lem05}
For $n\geq4$, the quantities $\varphi_n$, $\theta_n$, $\phi_n$,  and $\tau_n$, describing the number of  different maximum matchings in Apollonian network $\mathcal{A}_n$ under certain conditions, obey the following relation:
\begin{equation}
\varphi_n\leq\theta_n\leq\phi_n\leq\tau_n.
\end{equation}
\end{lemma}
\begin{proof}
We prove this lemma by induction on $n$.  Table~\ref{MatchingNo}  shows  that for $n=4$, the result holds. Let us assume that  the lemma is true for $n=k$, that is, $\varphi_k\leq\theta_k\leq\phi_k\leq\tau_k$. For $n=k+1$, according to induction assumption and  Theorem~\ref{Thm02},  we obtain $\varphi_{k+1}\leq\theta_{k+1}\leq\phi_{k+1}\leq\tau_{k+1}$. This completes the proof.
\end{proof}

In order to estimate the asymptotic growth constant $Z(\mathcal{A}_n)$,  we define several ratios: $\alpha_n=\frac{\varphi_n}{\theta_n}$, $\lambda_n=\frac{\tau_n}{\phi_n}$, and $\eta_n=\frac{\theta_n}{\phi_n}$. 
\begin{lemma}\label{lem06}
Let $\eta_0=\frac{2}{3}$ and $\lambda_0= \frac{12}{11}$. Then for $ n\geq3$, the three ratios $\alpha_n$, $\lambda_n$, and $\eta_n$ obey following relations:
$\alpha_n\downarrow0$,  $\eta_n\leq \eta_0=\frac{2}{3}$,  and $\lambda_n\geq \lambda_0= \frac{12}{11}$.
\end{lemma}
\begin{proof}
It is easy to see that $\alpha_n, \lambda_n, \eta_n$ are all constantly positive. By definition,
\begin{equation}
\alpha_{n+1}=\frac{\varphi_{n+1}}{\theta_{n+1}}=\frac{3\theta_n\varphi_n^2}{2\phi_n\varphi_n^2+4\theta_n^2\varphi_n}\leq\frac{3}{4}\cdot\frac{\varphi_n}{\theta_n}=\frac{3}{4}\alpha_n. \nonumber
\end{equation}
Thus,  $\alpha_n$ is a decreasing positive sequence that converges to zero, implying  $\alpha_n\downarrow0$ as $n\rightarrow\infty$.


Next we bound $\eta_n$ and $\lambda_n$. By definitions, we have
\begin{equation}
\eta_{n+1}=\frac{\theta_{n+1}}{\phi_{n+1}}= \frac{2 \phi_n\varphi_n^2+4\theta_n^2\varphi_n}{\tau_n\varphi_n^2+8\phi_n\theta_n\varphi_n+3\theta_n^3}
\leq\frac{2 \phi_n\theta_n\varphi_n+2\phi_n\theta_n\varphi_n+2\theta_n^3}{8\phi_n\theta_n\varphi_n+3\theta_n^3}
\leq\frac{2}{3} \nonumber
\end{equation}
and
\begin{align}
\lambda_{n+1}=\frac{\tau_{n+1}}{\phi_{n+1}}= & \frac{6\tau_n\theta_n\varphi_n+6\varphi_n\phi_n^2+12\theta_n^2\phi_n}{\tau_n\varphi_n^2+8\phi_n\theta_n\varphi_n+3\theta_n^3} \geq \frac{6\tau_n\theta_n\varphi_n+6\varphi_n\phi_n^2+12\theta_n^2\phi_n}{\tau_n\theta_n\varphi_n+8\phi_n\theta_n^2+3\theta_n^2\phi_n} \nonumber\\
\geq & \frac{6\tau_n\theta_n\varphi_n+12\theta_n^2\phi_n}{\tau_n\theta_n\varphi_n+11\theta_n^2\phi_n}
\geq  \frac{12}{11}, \nonumber
\end{align}
which, together with  the initial values $\eta_3\approx0.5125\leq\frac{2}{3}$ and $\lambda_3\approx1.5375\geq\frac{12}{11}$,  show that   $\eta_n\leq \eta_0=\frac{2}{3}$,  and $\lambda_n\geq \lambda_0= \frac{12}{11}$.
\end{proof}




The above properties of related ratios  are given in a loose way. One can  provide a tighter bound for $\lambda_n$, for example.  However, they are enough to find the asymptotic growth constant $Z(\mathcal{A})$.
\begin{theo}\label{Thm03}
For Apollonian network $\mathcal{A}_n$,  $Z(\mathcal{A}_n)$ has a limit when $n$ is sufficiently large, and the asymptotic growth constant  for the number of maximum matchings is $Z(\mathcal{A}_n)=0.43017\ldots$.
\end{theo}

\begin{proof}
By Eqs.~\eqref{M10},~\eqref{M11}, ~\eqref{M12}, ~\eqref{M13}, we obtain
\begin{equation}\label{varphi01}
\varphi_{n+1}= \frac{3}{\alpha_n}\varphi_n^3
\end{equation}
and
\begin{equation}\label{tau01}
\tau_{n+1}= 6\tau_n^3 \left(\frac{\eta_n}{\lambda_n}\cdot\frac{\alpha_n\eta_n}{\lambda_n}+\frac{\alpha_n\eta_n}{\lambda_n^3}+2\left(\frac{\eta_n}{\lambda_n}\right)^2 \frac{1}{\lambda_n}\right).
\end{equation}
Denote $p_n=\frac{3}{\alpha_n}$ and $q=6\left(\frac{\alpha_3\eta_0^2}{\lambda_0^2}+\frac{\alpha_3\eta_0}{\lambda_0^3}+2\left(\frac{\eta_0^2}{\lambda_0^3}\right)\right) \geq 6\left(\frac{\alpha_n\eta_n^2}{\lambda_n^2}+\frac{\alpha_n\eta_n}{\lambda_n^3}+2\left(\frac{\eta_n^2}{\lambda_n^3}\right)\right)$.
According to  Lemma~\ref{lem06},  one has $p_n\geq p_{n-1}\geq \cdots \geq p_3=\frac{3}{\alpha_3}$. From  Eqs.~\eqref{varphi01} and~\eqref{tau01},  we obtain that  for  $n>m\geq3$,
\begin{equation}\label{varphi02}
\varphi_{n}\geq  p_m \varphi_{n-1}^3\geq p_m  p_m^3  \varphi_{n-2}^9\geq \cdots\geq p_m ^{\frac{3^{n-m}-1}{2}} \varphi_m^{3^{n-m}},
\end{equation}
\begin{equation}\label{tau02}
\tau_{n}\leq  \cdots\leq  q^{\frac{3^{n-m}-1}{2}} \tau_m^{3^{n-m}}.
\end{equation}
Taking logarithm on both sides of Eqs.~\eqref{varphi02} and~\eqref{tau02} gives rise to
\begin{equation}
\log{\varphi_n}  \geq 3^{n-m}\log{\varphi_m}+\frac{3^{n-m}-1}{2}\log{p_m},
\end{equation}
\begin{equation}
\log{\tau_n}  \leq 3^{n-m}\log{\tau_m}+\frac{3^{n-m}-1}{2}\log{q}.
\end{equation}
Notice that for $n\geq3$,  $\varphi_n\leq \tau_n$ as stated in Lemma~\ref{lem05}. By definition, $\lim_{n\rightarrow\infty}Z(\mathcal{A}_n)=\lim_{ n\rightarrow\infty}\frac{\ln {\tau_n}}{V_n}=\lim_{n\rightarrow\infty}{2\frac{\ln{\tau_n}}{3^n}}$. Then,  for  $n\geq 3$,
\begin{equation}\label{Cons}
2\frac{\log{\varphi_m}}{3^m}+\frac{\log{p_m}}{3^m}\leq  Z(\mathcal{A}_n)   \leq2\frac{\log{\tau_m}}{3^m}+\frac{\log{q}}{3^m}.
\end{equation}
Because $\frac{\tau_{n+1}}{\varphi_{n+1}}=2\frac{\tau_n}{\phi_n}+2\frac{\phi_n^2}{\theta_n\varphi_n}+4\frac{\theta_n\phi_n}{\varphi_n^2}\leq 8(\frac{\tau_n}{\varphi_n})^2$,  then $2\frac{\log{\tau_m}-\log{\varphi_m}}{3^m}\rightarrow0$ as $m\rightarrow\infty$.  Therefore, the  difference of the leftmost  and  rightmost sides in  Eq.~\eqref{Cons} converges to zero as $m\rightarrow\infty$, meaning that  $Z(\mathcal{A}_n)$ has a limit when  $n >m\rightarrow\infty$.

The convergence of the upper and lower bounds for $Z(\mathcal{A}_n)$ in Eq.~\eqref{Cons} is rapid. For example, when $m$ is 7, the difference between the upper and lower bounds is less than $10^{-2}$.  Note that when $n=7$,  the difference of the exact values for  the upper and lower bounds is approximately equal to $ 0.4\times 10^{-2}$,  implying that  both bounds are a good approximate of the limit.  In other words, even for small $m$, the upper and lower bounds in Eq.~\eqref{Cons} converge to the quoted value $Z(\mathcal{A}_n)=0.43017\ldots$. 
\end{proof}






\subsection{Matching number and the number of maximum matchings in extended Tower of Hanoi graphs }

We proceed to study the matching number and the number of maximum matchings in extended Tower of Hanoi graphs, the dual of Apollonian networks. We will show that  in  contrast to Apollonian networks, there always exists a perfect matching in extended Tower of Hanoi graphs.   Moreover,  we will determine the number of  perfect matchings in  extended Tower of Hanoi graphs.

\subsubsection{Matching number }

In the case without inducing confusion,   we employ the same notation as those for $\mathcal{A}_n$ studied in the preceding subsection.  Let $\beta_n$ stand for the  matching number of $\mathcal{S}^+_n$.  Let $\mathcal{U}$ be a subset of $\mathcal{V}$ for graph $\mathcal{G}=(\mathcal{V},\mathcal{E})$, we use $\mathcal{G} \setminus \mathcal{U}$ to denote  a subgraph of $\mathcal{G}$, which is obtained from $\mathcal{G}$ by deleting those vertices in  $\mathcal{U}$ and edges adjacent to any vertex in $\mathcal{U}$.  Then,  $\mathcal{S}^+_n \setminus \{s\}$ is isomorphic to  $\mathcal{H}_n$.

\begin{theo}\label{Thm11}
In the extended Tower of Hanoi graph $\mathcal{S}^+_n$,  $n \geq 2$, a perfect matching always exists. Thus, the matching number of $\mathcal{S}^+_n$  is
\begin{equation}\label{MD6}
\beta_n=\frac{V_n}{2}=\frac{3^n+1}{2}.
\end{equation}
\end{theo}
\begin{proof}
It suffices to show that the graph $\mathcal{S}^+_n$, $n \geq 2$,  has a perfect matching. Then the second result immediately follows from the fact that $V_n=3^n+1$.

We first prove the existence of a perfect matching in $\mathcal{S}^+_n$,  $n \geq 2$, by induction on $n$. When $n=2$, it is easy to check that a perfect matching  exists in $\mathcal{S}^+_2$. Thus the result holds for the base case. Assume that for $n=k  \geq 2$, there exists a perfect matching  in $\mathcal{S}^+_k$, which means that there exists at least a perfect matching  in  the three subgraphs $\mathcal{H}_k\setminus \{a_k\}$, $\mathcal{H}_k\setminus \{b_k\}$, and $\mathcal{H}_k\setminus \{c_k\}$ of $\mathcal{H}_k$. Then there exists  a perfect matching $\mathcal{M}$ in  $\mathcal{H}_{k+1}\setminus \{a_{k+1}\}$,  $\mathcal{H}_{k+1}\setminus \{b_{k+1}\}$, and  $\mathcal{H}_{k+1}\setminus \{c_{k+1}\}$, see Fig.~\ref{fig:11}, where $a^i_{k}$, $b^i_{k}$, and $c^i_{k}$ denote the extreme vertices of $\mathcal{H}^i_k$, $i=1,2,3$,   forming  $\mathcal{H}_{k+1}$. By adding to  $\mathcal{M}$ the  edge connected the special vertex $s$ and the vacant extreme vertex $a_{k+1}$,  $b_{k+1}$, or  $c_{k+1}$, leads to a perfect matching of  $\mathcal{S}^+_{k+1}$,  as shown in Fig.~\ref{fig:11}.
\end{proof}

\begin{figure}
\centering
 \includegraphics[width=0.25\textwidth]{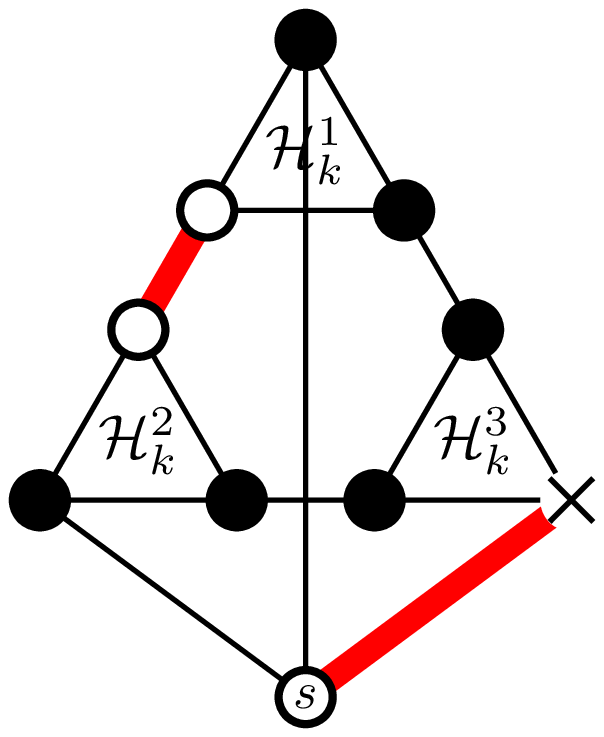}
\caption{\label{fig:11} Illustration of existence of a perfect matching in $\mathcal{S}^+_{k+1}$ for $k \geq 2$. The cross vertex denotes the removed extreme vertex $c_{k+1}$ from $\mathcal{H}_{k+1}$.  A perfect matching for $\mathcal{H}_{k+1}\setminus \{c_{k+1}\}$ can be obtained by adding one thick edges linked to the two extreme vertices $c^1_{k}$ in $\mathcal{H}^{1}_k$ and $c^2_{k}$  and $\mathcal{H}^{2}_k$ to the union of perfect matchings for  $\mathcal{H}^{1}_k\setminus \{c^1_{k}\}$,    $\mathcal{H}^{2}_k\setminus \{c^2_{k}\}$, and   $\mathcal{H}^{3}_k\setminus \{c^3_{k}\}$.  Adding to a perfect matching of  $\mathcal{H}_{k+1}\setminus \{c_{k+1}\}$ one thick edge connecting to the  extreme vertex $c_{k+1}$  and the special vertex $s$ gives a perfect matching for  $\mathcal{S}^+_{k+1}$.}
\end{figure}


\subsubsection{Number of perfect matchings}


Let $\beta^k_n$, $k=0,1,2,3$, be the maximum size of matchings of the subgraph of $\mathcal{H}_n$ obtained from $\mathcal{H}_n$ by deleting  exactly $k$ extreme vertices and the edges attaching to these  $k$ extreme vertices. For any matching $\mathcal{M}$ of $\mathcal{H}_n$ with at least one vacant extreme vertex, we can obtain a matching $\mathcal{M}'$ of  $\mathcal{S}^+_n$ from $\mathcal{M}$ by adding to it  an edge adjacent to a vacant extreme vertex and the special vertex $s$.

Note that for any perfect matching $\mathcal{M}'$ in $\mathcal{S}^+_{n+1}$, the number of matching edges in its subgraph $\mathcal{H}^i_n$, $i=1,2,3$, must be  $\beta^0_n$ or  $\beta^2_n$. Otherwise, we suppose that in some $\mathcal{H}^i_n$, the number of matching edges is $\beta^1_n$ or  $\beta^3_n$, then there must be one vacant non-extreme vertex in  $\mathcal{H}^i_n$,  meaning that $\mathcal{M}'$ cannot be a perfect matching  of $\mathcal{S}^+_{n+1}$.


Next, we show that in graph $\mathcal{H}_n\setminus \{a_n,b_n,c_n\}$, any matching with size $\beta^0_n$ is indeed  a perfect matching.
\begin{lemma}\label{lem11}
For $n \geq 2$, there exists a perfect matching in the subgraph $\mathcal{H}_n\setminus \{a_n,b_n,c_n\}$ of the extended Tower of Hanoi graph $\mathcal{H}_n$, and its matching number is
\begin{equation}\label{MD6}
\beta^0_n=\frac{3^n-3}{2}.
\end{equation}
\end{lemma}
\begin{proof}
We prove this lemma  by induction on $n$. When $n=2$,  the three edges  connecting the three pairs of extreme vertices in $\mathcal{H}^i_1$, $i=1,2,3$, form a  perfect matching of $\mathcal{H}_2\setminus \{a_2,b_2,c_2\}$.
Thus,  the base case  holds.  Suppose that for $n=k\ge 2$, there exists a perfect matching in $\mathcal{H}_k\setminus \{a_k,b_k,c_k\}$. Then for $n=k+1$, we can obtain a  perfect matching of $\mathcal{H}_{k+1}\setminus \{a_{k+1},b_{k+1},c_{k+1}\}$ by adding to the perfect matchings of $\mathcal{H}^i_k\setminus \{a^i_{k},b^i_{k},c^i_{k}\}$, $i=1,2,3$, three  edges connecting three pairs of extreme  vertices in $\mathcal{H}^i_n$,  See Fig.~\ref{fig:12}. Then, the number of edges in a perfect matching of $\mathcal{H}_n\setminus \{a_n,b_n,c_n\}$ is $\beta^0_n=\frac{3^n-3}{2}$.
\end{proof}

\begin{figure}
\centering
  \includegraphics[width=0.25\textwidth]{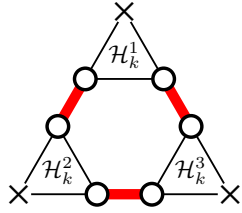}
\caption{\label{fig:12} Illustration of existence of a perfect matching in $\mathcal{H}_{k+1}\setminus \{a_{k+1},b_{k+1},c_{k+1}\}$ for $k\ge 2$. The cross vertices denote those removed extreme vertices from $\mathcal{H}_{k+1}$.}
\end{figure}

Let $\varphi_n$ be the number of maximum matchings of $\mathcal{H}_n\setminus \{a_n,b_n,c_n\}$. Let $\phi_n$ be the number of maximum matchings of  $\mathcal{H}_n\setminus\{a_n\}$, which  is equal to the  number of maximum matchings of $\mathcal{H}_n\setminus\{b_n\}$ or $\mathcal{H}_n\setminus\{c_n\}$.



\begin{lemma}\label{lem15}
For the  Tower of Hanoi graph $\mathcal{H}_n$, $n \geq2$,  the quantities $\varphi_n$ and $\phi_n$ satisfy the following relations:
\begin{equation}\label{MD7}
\varphi_{n+1} = \varphi_n^3+\phi_n^3,
\end{equation}
\begin{equation}\label{MD9}
\phi_{n+1} = \phi_n^2\varphi_n+\phi_n^3.
\end{equation}
\end{lemma}
\begin{proof}
We prove this lemma graphically. Fig.~\ref{fig:13} provides all configurations of maximum matchings in  $\mathcal{H}_{n+1}\setminus \{a_{n+1},b_{n+1},c_{n+1}\}$  that contribute to $\varphi_{n+1} $. Columns 1 and 4 in Fig.~\ref{fig:14} provide all configurations of maximum matchings in  $\mathcal{H}_{n+1} \setminus \{a_{n+1}\}$  that contribute to $\phi_{n+1}$.

Here we only given an explanation for  Eq.~\eqref{MD7},  while  Eq.~\eqref{MD9} can be accounted for analogously.

Fig.~\ref{fig:13} shows that in  graph $\mathcal{H}_{n+1}\setminus \{a_{n+1},b_{n+1},c_{n+1}\}$, all possible configurations of maximum matchings with size $\beta^0_{n+1}$ are those having size $3\beta^0_n+3$ or $3\beta^2_n$. From Fig.~\ref{fig:13}, we can see that the number of all maximum matchings with size $3\beta^0_n+3$ and $3\beta^2_n$ is equal to $\varphi_n^3$ and $\phi_n^3$, respectively. By definition of $\varphi_{n+1}$,  Eq.~\eqref{MD7} holds.
\end{proof}

\begin{figure}
\centering
  \includegraphics[width=0.4\textwidth]{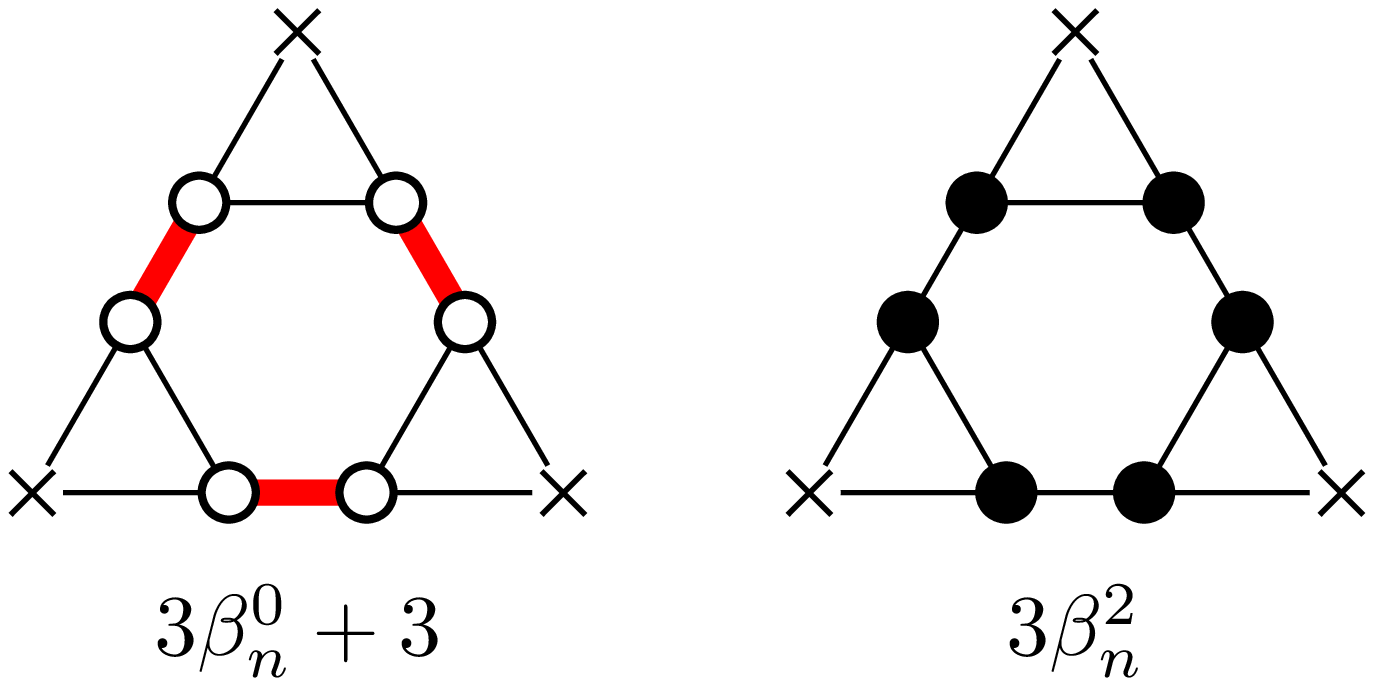}
\caption{\label{fig:13} Configurations of maximum matchings for $\mathcal{H}_{n+1}\setminus \{a_{n+1},b_{n+1},c_{n+1}\}$ .}
\end{figure}

\begin{figure}
\centering
   \includegraphics[width=0.85\textwidth]{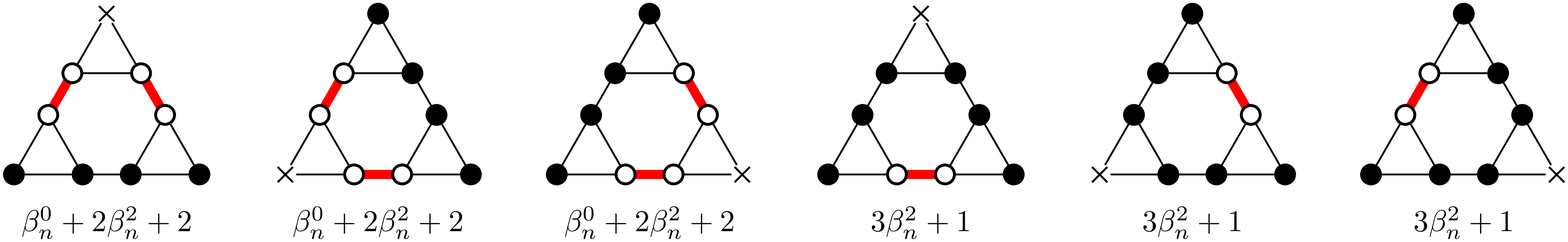}
\caption{\label{fig:14} Configurations of maximum matchings for $\mathcal{H}_{n+1}\setminus \{a_{n+1}\}$, $\mathcal{H}_{n+1}\setminus \{b_{n+1}\}$, and $\mathcal{H}_{n+1}\setminus \{c_{n+1}\}$.}
\end{figure}



The recursive relations stated in Lemma~\ref{lem15} allow to determine explicitly the quantities $\varphi_n$ and $\phi_n$.

\begin{lemma}\label{lem16}
For  the  Tower of Hanoi graph $\mathcal{H}_n$, $n \geq2$,
\begin{equation}\label{MD11}
\varphi_{n} = \phi_{n} = 2^{\frac{3^{n-1}-1}{2}}.
\end{equation}
\end{lemma}
\begin{proof}
Eqs.~\eqref{MD7} and~\eqref{MD9} implies $\varphi_{n+1}-\phi_{n+1}=\varphi_{n}({\varphi^2_{n}}-{\phi^2_{n}})$, which together with $\varphi_{2} = \phi_{2}=2$ shows that $\varphi_{n} = \phi_{n}$  for all $n\ge2$. Then, considering Eq.~\eqref{MD7}, we get the recursive relation for $\varphi_{n}$ as $\varphi_{n+1}=2{\varphi^3_{n}}$ that can be easily solved to yield $\varphi_{n} = \phi_{n} = 2^{\frac{3^{n-1}-1}{2}}$.
\end{proof}

\begin{theo}\label{Thm12}
The number of perfect matchings in the extended Tower of Hanoi graph $\mathcal{S}^+_n$, $n \geq2$, is $3\cdot2^{\frac{3^{n-1}-1}{2}}$.
\end{theo}
\begin{proof}
Because for any perfect matching of $\mathcal{S}^+_n$, the special vertex $s$ is covered simultaneously with an extreme vertex in $\mathcal{H}_n$, there is a one-to-one correspondence between perfect matchings of $\mathcal{S}^+_n$ and perfect matchings in $\mathcal{H}_n\setminus\{a_n\}$, $\mathcal{H}_n\setminus\{b_n\}$ or $\mathcal{H}_n\setminus\{c_n\}$.  Thus, the  number of perfect matchings in $\mathcal{S}^+_n$ is  $3\phi_{n}=3\cdot2^{\frac{3^{n-1}-1}{2}}$. 
\end{proof}

It should be mentioned that the problem of all matchings in the Tower of Hanoi graphs has been previously studied~\cite{ChWuHuDe15}. However, here we address the problem of perfect matching in  the Hanoi graphs  with one extreme vertex removed. Moreover, our technique is different from that in~\cite{ChWuHuDe15}, but is partially similar to those in~\cite{ChCh08} and~\cite{TeWa07}.

\section{Domination number and the number of minimum  dominating sets in Apollonian Networks and extended Tower of Hanoi graphs}

In this section, we study the domination number and the number of minimum dominating sets in Apollonian Networks and extended Tower of Hanoi graphs.



\subsection{Domination number and the number of minimum  dominating sets in Apollonian Networks}

We first address the domination number and the number of minimum  dominating sets in Apollonian networks, by using their structural self-similarity.

\subsubsection{Domination number}

Let $\gamma_n$ be the domination number of $\mathcal{A}_n$.  All dominating sets of $\mathcal{A}_n$ can be sorted into four classes: $\Theta_n^0$, $\Theta_n^1$, $\Theta_n^2$, and $\Theta_n^3$, where $\Theta_n^k$, $k = 0,1,2,3$, denotes the set of  those dominating sets, each containing exactly $k$ outmost vertices. Let $\Upsilon_n^k$, $k = 0,1,2,3$, be the subsets of $\Theta_n^k$, each of which has the smallest cardinality, denoted by $\gamma_n^k$.  The  dominating sets in $\Upsilon_{n+1}^k$, $k = 0,1,2,3$, can be further sorted into two groups: in one group, each dominating set contains the center vertex, while in the other group, every dominating set excludes the center vertex.

\begin{lemma}\label{lem22}
The domination number of Apollonian network $\mathcal{A}_n$, $n\geq3$, is
$\gamma_n = \min\{\gamma^0_n, \gamma^1_n, \gamma^2_n, \gamma^3_n\}$.
\end{lemma}

Thus, the problem of  determining $\gamma_n$ is reduced to evaluating $\gamma^0_n$, $\gamma^1_n$, $\gamma^2_n$, $\gamma^3_n$.  We next estimate these quantities.  
\begin{lemma}\label{lem23}
For  Apollonian network $\mathcal{A}_n$, $n\geq3$,
\begin{equation}\label{D1}
\gamma^0_{n+1}=\min\{3\gamma^0_n,3\gamma^1_n-2\},
\end{equation}
\begin{equation}\label{D2}
\gamma^1_{n+1}=\min\{2\gamma^1_n-1,2\gamma^2_n+\gamma^1_n-3\},
\end{equation}
\begin{equation}\label{D3}
\gamma^2_{n+1}=\min\{\gamma^2_n+2\gamma^1_n-2,\gamma^3_n+2\gamma^2_n-4\},
\end{equation}
\begin{equation}\label{D4}
\gamma^3_{n+1}=\min\{3\gamma^2_n-3,3\gamma^3_n-5\}.
\end{equation}
\end{lemma}
\begin{proof}
By definition, $\gamma_{n+1}^k$, $k = 0,1,2,3$, is the cardinality of a set in $\Upsilon_{n+1}^k$.  The four sets $\Upsilon_{n+1}^k$, $k = 0,1,2,3$, can be iteratively constructed from $\Upsilon_n^0$, $\Upsilon_n^1$, $\Upsilon_n^2$, and $\Upsilon_n^3$, as will be shown below. Then, $\gamma_{n+1}^k$, $k = 0,1,2,3$, can be derived from $\gamma_n^0$, $\gamma_n^1$, $\gamma_n^2$, and $\gamma_n^3$.

Eqs.~\eqref{D1}-\eqref{D4}  can all be proved graphically. We first prove Eq.~\eqref{D1}.

By Definition~\ref{Def:AP02},   $\mathcal{A}_{n+1}$  consists  of three copies of $\mathcal{A}_n$, $\mathcal{A}_{n}^{i}$ with $i=1,2,3$. By definition, for an arbitrary dominating set $\mathcal{D}$ belonging to $\Upsilon_{n+1}^0$, the three outmost vertices of $\mathcal{A}_{n+1}$ are not in $\mathcal{D}$, implying that the outmost vertices $Y_{i}$ and $Z_{i}$ of $\mathcal{A}_{n}^{i}$, $i=1,2,3$, do not belong to $\mathcal{D}$, see Fig.~\ref{FigApo02}. Then, we can construct $\mathcal{D}$ from $\Upsilon_n^0$, $\Upsilon_n^1$, $\Upsilon_n^2$, and $\Upsilon_n^3$ by considering whether the outmost vertices $X_{i}$ of $\mathcal{A}_{n}^{i}$, $i=1,2,3$, are in $\mathcal{D}$ or not.  Fig.~\ref{fig:21} illustrates all possible configurations of dominating sets in $\Theta_{n+1}^0$ that contains all dominating sets in $\Upsilon_{n+1}^0$. In Fig.~\ref{fig:21}, only the outmost vertices of $\mathcal{A}_{n}^{i}$,  $i=1,2,3$, are shown, with solid vertices being in the dominating sets, while open vertices not. From Fig.~\ref{fig:21}, we obtain
\begin{equation*}
\gamma^0_{n+1}=\min\{3\gamma^0_n,3\gamma^1_n-2\}.
\end{equation*}

Similarly, we can prove Eqs.~\eqref{D2},~\eqref{D3},  and~\eqref{D4} . In Figs.~\ref{fig:22},~\ref{fig:23}, and~\ref{fig:24},  we provide graphical representations of them.
\end{proof}


\begin{figure}
\centering
    \includegraphics[width=0.3\textwidth]{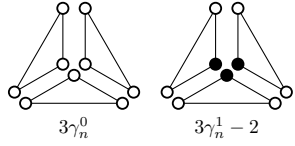}
	\caption{\label{fig:21} All possible configurations of dominating sets in  $\Theta_{n+1}^0$ for $\mathcal{A}_{n+1}$,  which contains the elements in $\Upsilon_{n+1}^0$ . }
\end{figure}
\begin{figure}
\centering
    \includegraphics[width=0.85\textwidth]{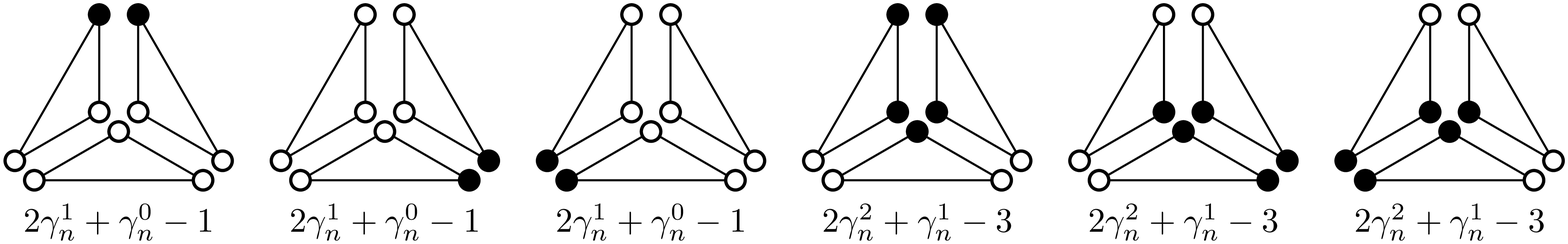}
	\caption{\label{fig:22} All possible configurations of dominating sets in  $\Theta_{n+1}^1$ for $\mathcal{A}_{n+1}$,  which contains  the elements in $\Upsilon_{n+1}^1$ . }
\end{figure}
\begin{figure}
\centering
    \includegraphics[width=0.85\textwidth]{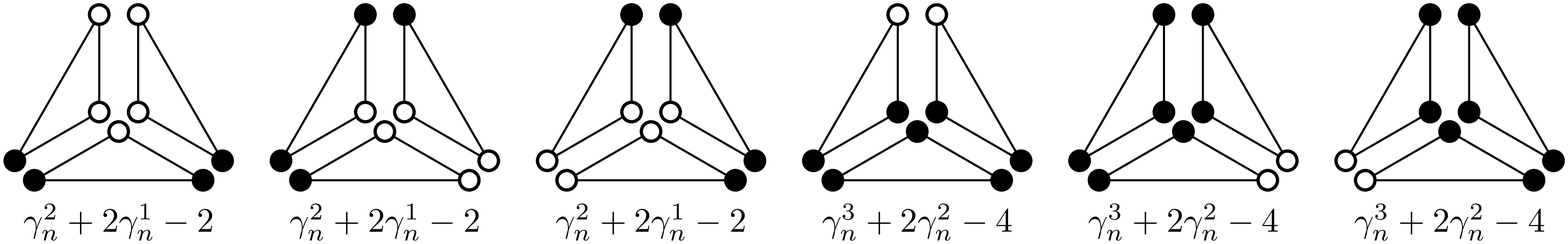}
	\caption{\label{fig:23} All possible configurations of dominating sets in  $\Theta_{n+1}^2$ for $\mathcal{A}_{n+1}$,  which contains the elements in $\Upsilon_{n+1}^2$ . }
\end{figure}
\begin{figure}
\centering
    \includegraphics[width=0.3\textwidth]{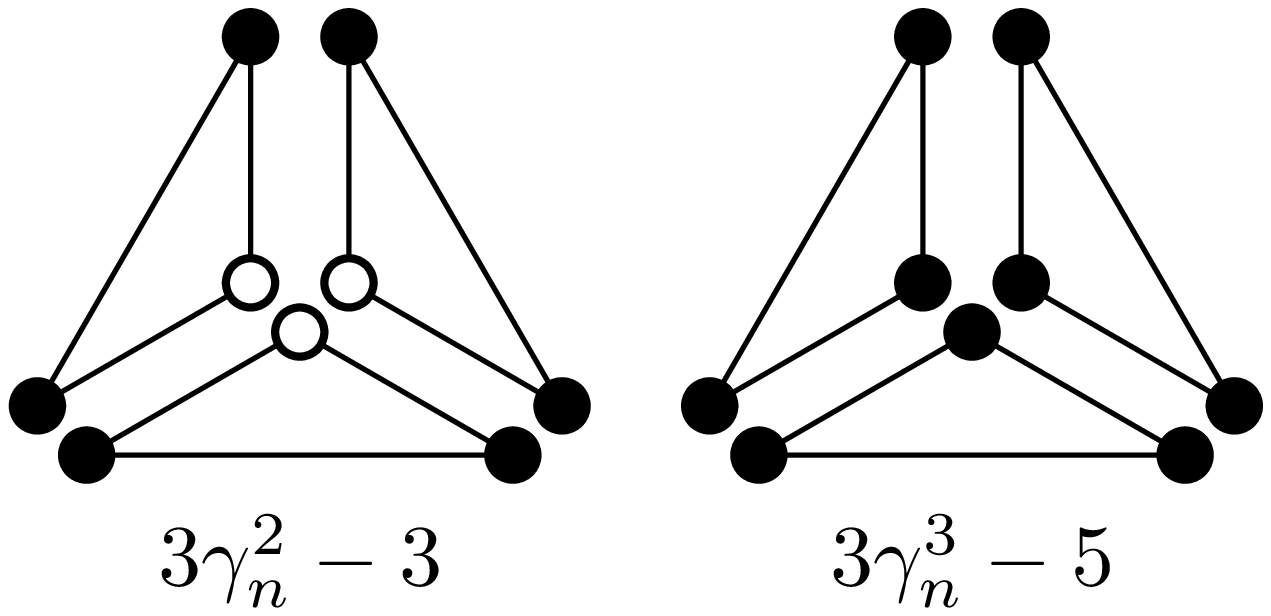}
	\caption{\label{fig:24} All possible configurations of dominating sets in  $\Theta_{n+1}^3$ for $\mathcal{A}_{n+1}$,  which contains the elements in $\Upsilon_{n+1}^3$.}
\end{figure}

For small $n$, $\gamma_n^0$, $\gamma_n^1$, $\gamma_n^2$, and $\gamma_n^3$ can be determined by hand.  For example,  for  $\mathcal{A}_n$, $n=1,2,3$, we have $(\gamma^0_1,\gamma^1_1,\gamma^2_1,\gamma^3_1)=(1,1,2,3)$, $(\gamma^0_2,\gamma^1_2,\gamma^2_2,\gamma^3_2)=(1,2,2,3)$, and $(\gamma^0_3,\gamma^1_3,\gamma^2_3,\gamma^3_3)=(3,3,3,3)$.

For $n \geq 3$, these quantities follow some relation as stated in the following lemma.
\begin{lemma}\label{lem25}
For Apollonian network $\mathcal{A}_n$, $n \geq 3$, the four quantities $\gamma_n^0$, $\gamma_n^1$, $\gamma_n^2$, and $\gamma_n^3$ obey the following relations:
\begin{equation}\label{D5}
\gamma^0_n\geq\gamma^1_n\geq\gamma^2_n\geq\gamma^3_n.
\end{equation}
\end{lemma}
\begin{proof}
We  prove this lemma by induction on $n$.  When $n=3$, $\gamma^0_3=\gamma^1_3=\gamma^2_3=\gamma^3_3=3$. Thus, the basis step holds.

Let us suppose that the relation holds true for $n=k$, that is, $\gamma^0_k\geq \gamma^1_k\geq \gamma^2_k\geq \gamma^3_k$.  Then, from Eq.~\eqref{D1}, we obtain
\begin{equation*}
\gamma^0_{k+1}=\min\{3\gamma^0_k,3\gamma^1_k-2\}=3\gamma^1_{k}-2.
\end{equation*}
Analogously, we can derive
$$
\begin{aligned}
\gamma^1_{k+1}= & 2\gamma^2_{k}+\gamma^1_{k}-3,\\
\gamma^2_{k+1}= & \gamma^3_{k}+2\gamma^2_{k}-4,\\
\gamma^3_{k+1}= & 3\gamma^3_{k}-5.
\end{aligned}
$$
Comparing the above-obtained four relations and using induction hypothesis $\gamma^0_k\geq \gamma^1_k\geq \gamma^2_k\geq \gamma^3_k$, we arrive at $\gamma^0_{k+1}\geq \gamma^1_{k+1}\geq \gamma^2_{k+1}\geq \gamma^3_{k+1}$. This completes the proof.
\end{proof}


\begin{theo}\label{Thm21}
The domination number of Apollonian network $\mathcal{A}_n$, $n \geq 3$,  is
\begin{equation}\label{D6}
\gamma_n=\frac{3^{n-3}+5}{2}.
\end{equation}
\end{theo}
\begin{proof}
Lemmas~\ref{lem22} and~\ref{lem25} imply  that $\gamma_n=\gamma^3_n$.  Then, we have the following recursion relation for
\begin{equation}\label{Fgamma05}
\gamma_{n+1}=\gamma^3_{n+1}=3 \gamma^3_n-5=3 \gamma_n-5.
\end{equation}
With the initial value $\gamma_3=3$, this relation is solved to  yield $\gamma_n=\frac{3^{n-3}+5}{2}$.
\end{proof}

Therefore, in a large Apollonian network $\mathcal{A}_n$,  $n\rightarrow\infty$, $\frac{\gamma_n}{V_n}=\frac{3^{n-3}+5}{3^n+5}\approx  \frac{1}{27}$.

\begin{coro}\label{lem26}
For Apollonian network $\mathcal{A}_n$, $n\geq3$,
\begin{equation}\label{D7}
\gamma^0_{n}=\frac{3^{n-3}+3\cdot2^{n-2}-1}{2},
\end{equation}
\begin{equation}\label{D8}
\gamma^1_{n}=\frac{3^{n-3}+2^{n-1}+1}{2},
\end{equation}
\begin{equation}\label{D9}
\gamma^2_{n}=\frac{3^{n-3}+2^{n-2}+3}{2}.
\end{equation}
\end{coro}
\begin{proof}
By Lemmas~\ref{lem23}  and~\ref{lem25} and Eq.~\eqref{D6}, the result is concluded immediately.
\end{proof}

\subsubsection{Number of minimum dominating sets}

Let $w_n$ be the number of minimum dominating sets in Apollonian network $\mathcal{A}_n$.
\begin{theo}\label{Thm22}
The Apollonian network $\mathcal{A}_n$, $n\geq4$, has a unique minimum domination set, that is, $w_n=1$.
\end{theo}
\begin{proof}
Equation~\eqref{Fgamma05} and Fig.~\ref{fig:24} indicate that for $n \geq 4$, any minimum dominating set of $\mathcal{A}_{n+1}$ is the union of three minimum dominating sets in $\Upsilon_n^3$, of the three copies of $\mathcal{A}_{n}$ (i.e., $\mathcal{A}_{n}^{1}$, $\mathcal{A}_{n}^{2}$, and $\mathcal{A}_{n}^{3}$) forming $\mathcal{A}_{n+1}$, with their identified outmost vertices being counted only four times. Since the minimum dominating set of $\mathcal{A}_4$ is unique, for $n \geq 4$, there is a unique  minimum dominating set in $\mathcal{A}_{n}$. In fact, the unique dominating set of $\mathcal{A}_n$, $n \geq 4$, is  the set of all vertices of $\mathcal{A}_{n-2}$.
\end{proof}


The Apollonian network $\mathcal{A}_n$ has some interesting properties. For example, a dominating set of any of its subgraphs, say $\mathcal{A}'_n$, obtained from $\mathcal{A}_n$ by deleting one or several outmost vertices is also a dominating set of the whole network $\mathcal{A}_n$. This can be explained as follows. For a new neighboring vertex (with degree 3) of an outmost vertex in $\mathcal{A}_n$ not in its subgraph $\mathcal{A}'_n$, it must be generated at iteration $n$ and its other neighbors are also neighbors of the outmost vertex. Thus, any dominating set of the subgraph $\mathcal{A}'_n$ must simultaneously dominate the new vertex and its neighboring outmost vertex. Therefore, any dominating set of $\mathcal{A}'_n$ is also a dominating set of $\mathcal{A}_n$.


In addition to the number of minimum dominating sets in Apollonian network $\mathcal{A}_n$ itself, the number of minimum dominating sets in some of its subgraphs can also be determined explicitly. Let $x_n$ be the number of minimum dominating sets of $\mathcal{A}_n\setminus \{X,Y,Z\}$, let $y_n$ be the number of minimum dominating sets of $\mathcal{A}_n\setminus \{X,Y\}$, which is equal to the number of minimum dominating sets of $\mathcal{A}_n\setminus \{Y,Z\}$ or $\mathcal{A}_n\setminus \{Z,X\}$, and let $z_n$ be the number of minimum dominating sets of $\mathcal{A}_n\setminus \{X\}$, which equals the number of dominating sets of $\mathcal{A}_n\setminus \{Y\}$ or $\mathcal{A}_n\setminus \{Z\}$.


\begin{lemma}\label{lem29}
For the Apollonian network  $\mathcal{A}_n$, $n\geq4$, $x_n=8$, $y_n=2$, and $z_n=1$.
\end{lemma}
\begin{proof}
Figs.~\ref{fig:21},~\ref{fig:22}, and~\ref{fig:23} show that the quantities $x_n$, $y_n$, and $z_n$  satisfy the following relations:
\begin{equation}\label{D10}
x_{n+1} = y_n^3,
\end{equation}
\begin{equation}\label{D11}
y_{n+1} = z_n^2y_n,
\end{equation}
and
\begin{equation}\label{D12}
z_{n+1} = w_nz_n^2.
\end{equation}
Using $x_3=1$, $y_3=2$, $z_3=1$, and $w_n=1$ for all $n\geq 3$, the above-obtained relations in Eqs.~\eqref{D10},~\eqref{D11}, and~\eqref{D12} are solved to give the result.
\end{proof}



\subsection{Domination  number and the number of minimum dominating sets   in extended Tower of Hanoi graphs}

We are now in a position to study the domination  number and the number of minimum dominating sets   in extended Tower of Hanoi graphs.

\subsubsection{Domination  number}

Let $\gamma_n$ denote  the domination number of the  extended Tower of Hanoi graph $\mathcal{S}^+_n$, $n  \geq 1$.
\begin{theo}\label{Thm31}
The domination number of the  extended Tower of Hanoi graph $\mathcal{S}^+_n$, $n  \geq 1$, is
$$  \gamma_n=\left\{
\begin{aligned}
\frac{3^n+1}{4},\ & n\ odd,\\
\frac{3^n+3}{4},\ & n\ even.
\end{aligned}
\right.
$$
\end{theo}
\begin{proof}
The proof  technique is similar to that in~\cite{LiNe98}.  The basic idea  is to construct a minimum dominating set $\mathcal{D}$ of $\mathcal{H}_n$ ($n  \geq 1$) with size $\gamma_n$,  and then show that $\mathcal{D}$ is also a minimum dominating set of $\mathcal{S}^+_n$. 

For a vertex $\xi=\xi_1 \xi_2 \ldots \xi_{n-1} \xi_n$ in $\mathcal{H}_n$, Let $f^{k}_{\xi}$ denote the number of terms in the label for $\xi$ that are equal to $k$ for $k=0$, 1, or 2.    Define  $\mathcal{D}$ as a subset of  vertices in  $\mathcal{H}_n$ such that
 $\mathcal{D}=\{\xi \in \mathcal{V}(\mathcal{H}_n) | f^{0}_{\xi} \, {\rm and} \,  f^{1}_{\xi}\,    {\rm are\, even} \}$.  Then, for even $n$, all the three extreme vertices are in $\mathcal{D}$; while for odd $n$, only extreme vertex $c_n$ is in $\mathcal{D}$. By definition,  for any pair of nearest vertices in  $\mathcal{D}$, their distance is exactly 3, thus they have no neighbors in common~\cite{LiNe98}. Moreover, any vertex in $\mathcal{D}$ has 3 neighbors, except the extreme vertices, each having only two neighbors in  $\mathcal{H}_n$.  Note that any other vertex in $\mathcal{V}(\mathcal{H}_n) \setminus \mathcal{D}$ is adjacent to exactly one  vertex in $\mathcal{D}$~\cite{LiNe98}.  Therefore, the number of vertices in $\mathcal{D}$ is $| \mathcal{D}|=\frac{3^n+3}{4}$ for even $n$ or $| \mathcal{D}|=\frac{3^n+1}{4}$ for odd $n$. Hence $\mathcal{D}$ is a dominating set of $\mathcal{H}_n$.

We proceed to show that $\mathcal{D}$ is also a minimum dominating set  of  $\mathcal{H}_n$.  Since all vertices of $\mathcal{H}_n$  have a degree 3, except the three extreme vertices (each with a degree 2), every vertex in a minimum dominating set can dominate at most 4 vertices including itself. As a result,  any dominating set of  $\mathcal{H}_n$ contains at least $\lceil\frac{3^n}{4}\rceil=\frac{3^n+1}{4}$ vertices  for even $n$, or $\lceil\frac{3^n}{4}\rceil=\frac{3^n+3}{4}$  vertices for odd $n$.  Thus,  $\mathcal{D}$  is actually a minimum dominating set  of  $\mathcal{H}_n$.  Considering  that $\mathcal{D}$ includes at least one extreme vertex,  which is a neighbor of the special vertex $s$ in $\mathcal{S}^+_n$, which implies  that   $\mathcal{D}$ is  a  dominating set  of $\mathcal{S}^+_n$.  Because in any  dominating set  of  $\mathcal{S}^+_n$, there are  at least $\lceil\frac{3^n+1}{4}\rceil=\frac{3^n+1}{4}$ vertices for even $n$, or $\lceil\frac{3^n+1}{4}\rceil=\frac{3^n+3}{4}$ vertices for odd $n$,   $\mathcal{D}$ is also a minimum dominating set  of $\mathcal{S}^+_n$. This completes the proof.
\end{proof}

Theorem~\ref{Thm31} shows  that  the domination number of the  extended Tower of Hanoi graph $\mathcal{S}^+_n$, $n  \geq 1$, is equivalent to that of the corresponding  Tower of Hanoi graph $\mathcal{H}_n$~\cite{LiNe98,CuNe99,KlMiPe02,GrKlMo05}.

\subsubsection{Number of minimum dominating sets in $\mathcal{S}^+_n$ for odd $n$ }

We will show that except for the  domination number, the number of minimum dominating sets in the  extended Tower of Hanoi graph $\mathcal{S}^+_n$, $n  \geq 1$, can also be determined explicitly, at least for odd $n$.
 Hereafter,  we will only consider the  number of minimum dominating sets in $\mathcal{S}^+_n$ for odd $n$.

Given a minimum dominating  set $\mathcal{D}$ of $\mathcal{S}^+_n$,  the  vertices of any subset $\Delta \subset \mathcal{V}(\mathcal{S}^+_n)$   can be classified into three types. For any vertex  $y   \in  \Delta$,  either  $y$ is in  $ \mathcal{D}$, or $y$ is adjacent to  a vertex in $\mathcal{D}$ that is a disjoint union of $\mathcal{D}\cap \Delta$ and $\mathcal{D} \setminus \Delta$. If  $y   \in  \mathcal{D}$,  we say $y$ is  of type I; if  $y$ is adjacent to a vertex in $\mathcal{D}\cap \Delta$, then $y$ is  of type D; otherwise $y$ is  adjacent to a vertex in $\mathcal{D} \setminus \Delta$,  then $y$ is  of type C.  According to this definition,   for an arbitrary  minimum domination set  $\mathcal{D}$ of $\mathcal{S}^+_n$, if we set $\Delta= \mathcal{V}(\mathcal{H}_n)$, then  the possible types of the three extreme vertices $a_n$, $b_n$, and $c_n$ of   $\mathcal{H}_n$ are  in turn  D-D-I, D-I-D, I-D-D, or C-C-C.

Note that $\mathcal{H}_{n+2}$  consists of nine copies of $\mathcal{H}_{n}$, denoted as $\mathcal{H}_{n}^{i}$, $i=1,2,\ldots,9$. For any minimum dominating set $\mathcal{D}_2$ of $\mathcal{S}^+_{n+2}$,  which is obtained from  $\mathcal{H}_{n+2}$  by adding a special vertex $s$ connected to the three extreme vertices of   $\mathcal{H}_{n+2}$,
if we set   $\Delta= \mathcal{V}(\mathcal{H}_{n}^{i})$,  then with respect to  $\mathcal{D}_2$, the corresponding three extreme vertices  in  $\mathcal{H}_{n}^{i}$ could possibly belong to one of three  types:  (i) C-C-C,  (ii)  I-D-D,  (iii)  I-I-C, because any vertex in $\mathcal{D}_2$ dominates exactly four vertices including itself,  each without being dominated by  any other vertex in  $\mathcal{D}_2$.  However, the type  I-I-C actually does not exist as stated in the following lemma.

\begin{figure}
\centering
 \includegraphics[width=0.45\textwidth]{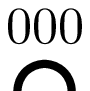}
\caption{\label{fig:31}Illustration of all four minimum dominating sets and four different disjoint classes of vertices or in $\mathcal{S}^+_{3}$. The vertices with the same color form a class or a  minimum dominating set.}
\end{figure}

 \begin{lemma}\label{lem42}
For any minimum domination set  of $\mathcal{S}^+_{n+2}$  with  odd  $n\geq 1$,  the type  I-I-C for  the three extreme vertices  in  any $\mathcal{H}_{n}^{i}$, $i=1,2,\ldots,9$,  is not possible to exist.
\end{lemma}
\begin{proof}
Let $\mathcal{D}_2$ be a minimum dominating set of $\mathcal{S}^+_{n+2}$.  As shown above,  if we restrict $\Delta= \mathcal{V}(\mathcal{H}_{n+2})$, then for the three  extreme  vertices of  $\mathcal{H}_{n+2}$, either all of them are of type C, or  one extreme vertex   is of type I, while the other two are of type D.  
For any pair of adjacent extreme  vertices  belonging to two different copies $\mathcal{H}_{n}^{i}$, their types under $\Delta= \mathcal{V}(\mathcal{H}_{n}^{i})$ obey the following two rules: (1) Type C is adjacent to only type I, and vice versa; (2) D is adjacent to only D.
According to these rules, by induction we can prove that in $\mathcal{S}^+_{n+2}$, the type  (iii)  I-I-C for  the three extreme vertices  in  any $\mathcal{H}_{n}^{i}$ actually doesn't exist.

For $n=1$, it is easy to check by hand that there are only four minimum dominating sets in $\mathcal{S}^+_{3}$, see Fig.~\ref{fig:31}, where  type I-I-C is not existent. Assume that  the result holds for $n-2$ ($n \geq 3$), i.e. for any minimum dominating set in $\mathcal{S}^+_{n}$  there is no type I-I-C  in any  $\mathcal{H}_{n-2}^{i}$, $i=1,2,\ldots,9$,   when setting $\Delta= \mathcal{V}(\mathcal{H}_{n-2}^{i})$.  For $n$, by contradiction, suppose that  type  I-I-C exists for the three extreme vertices of a subgraph $\mathcal{H}_{n}^{i}$ of $\mathcal{S}^+_{n+2}$, Fig.~\ref{fig:TypeIIC} shows that the extreme vertex $\bigstar$ of  the topmost $\mathcal{H}_{n-2}^{i'}$ as a constituent of some $\mathcal{H}_{n}^{i}$ should be of type I. This is not possible to exist, since for this case, at least two extreme vertices in  the topmost $\mathcal{H}_{n-2}^{i'}$  are of type I, which is in contradiction to the induction hypothesis. This completes the lemma.
\end{proof}

\begin{figure}
\centering
\includegraphics[width=0.4\textwidth]{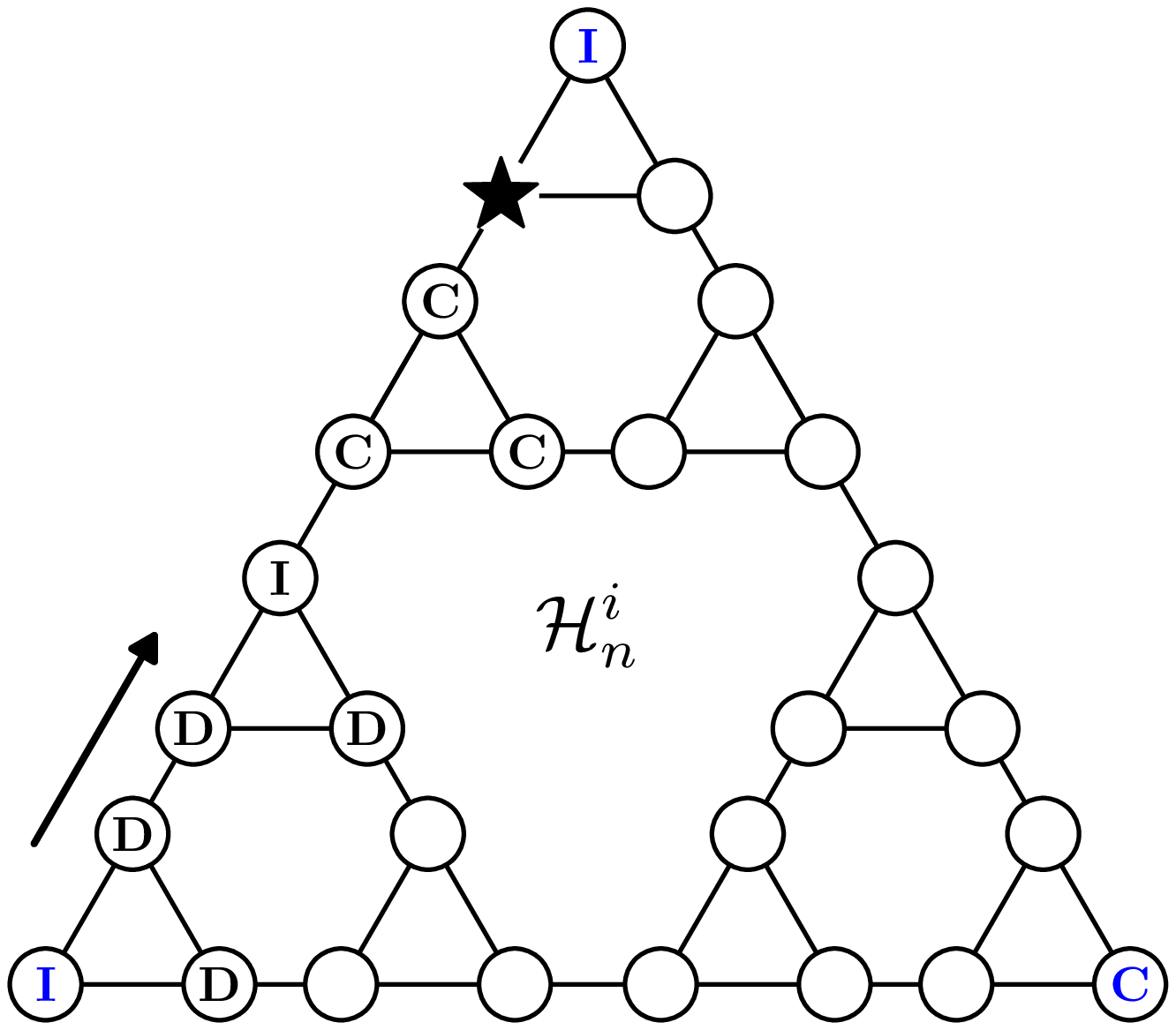}
\caption{\label{fig:TypeIIC} Types of extreme vertices for  a certain $\mathcal{H}_{n}^{i}$ as one of the nine constituents of   $\mathcal{H}_{n+2}$. Assume that the types of three extreme vertices are assumed to be I-I-C.  According to the rules governing the types of adjacent  extreme  vertices  belonging to two different copies  $\mathcal{H}_{n-2}^{i'}$, $i'=1,2,\ldots,9$, the extreme vertex $\bigstar$ of the topmost $\mathcal{H}_{n-2}^{i'}$ cannot have a legal type, when we label the vertex types along the arrow.}
\end{figure}

We now determine all minimum dominating sets in $\mathcal{S}^+_n$ for odd $n$.
Note that for  odd $n$, all vertices of  $\mathcal{S}^+_n$ can be categorized  into four disjoint classes:
$$ \mathcal{X}^1_n=\{\xi \in \mathcal{V}(\mathcal{H}_n) | f^{0}_{\xi} \, {\rm and} \,  f^{1}_{\xi}\,    {\rm are\, even} \},$$
$$\mathcal{X}^2_n=\{\xi \in \mathcal{V}(\mathcal{H}_n) | f^{0}_{\xi} \, {\rm and} \,  f^{2}_{\xi}\,    {\rm are\, even} \},$$
$$\mathcal{X}^3_n=\{\xi \in \mathcal{V}(\mathcal{H}_n) | f^{1}_{\xi} \, {\rm and} \,  f^{2}_{\xi}\,    {\rm are\, even} \},$$
$$\mathcal{X}^4_n=\{\xi \in \mathcal{V}(\mathcal{H}_n) | f^{0}_{\xi}, \,  f^{1}_{\xi}\,   {\rm and} \,    f^{2}_{\xi}\,{\rm are\, odd} \}\cup \{s\},$$
satisfying that $\mathcal{V}(\mathcal{S}^+_n)=\mathcal{X}^1_n\cup\mathcal{X}^2_n\cup \mathcal{X}^3_n \cup\mathcal{X}^4_n$ and $\mathcal{X}^1_n\cap \mathcal{X}^2_n\cap \mathcal{X}^3_n \cap\mathcal{X}^4_n= {\O} $.  Then $|\mathcal{X}^1_n |=| \mathcal{X}^2_n|=|\mathcal{X}^3_n|=|\mathcal{X}^4_n|=\frac{3^n+1}{4}$.   Fig.~\ref{fig:31} illustrates  the four distinct classes  of vertices in $\mathcal{S}^+_{3}$.

\begin{theo}\label{Thm32}
For each odd $n\geq 1$,  there are exactly  four distinct minimum dominating  sets, $\mathcal{X}^1_n$,  $\mathcal{X}^2_n$, $\mathcal{X}^2_n$, and $\mathcal{X}^4_n$, in the extended Tower of Hanoi graph $\mathcal{S}^+_n$.
\end{theo}
\begin{proof}
We first prove that each $\mathcal{X}^k_n$, $k=1,2,3,4$,  is  a  minimum dominating  set of $\mathcal{S}^+_n$, following a similar approach  in~\cite{LiNe98}.
By definition,  the  distance between any pair of nearest vertices in  $\mathcal{X}^k_n$  is 3,  implying  any two vertices in $\mathcal{X}^k_n$ have no common neighbors~\cite{LiNe98}. In addition, the degree of any vertex in $\mathcal{X}^k_n$ is 3 with each neighbor belonging to three different subsets $\mathcal{X}^{k'}_n$ other than $\mathcal{X}^k_n$, and any vertex in $\mathcal{V}(\mathcal{H}_n) \setminus \mathcal{X}^k_n$ is linked to exactly one  vertex in $\mathcal{X}^k_n$.  Thus,  every  $\mathcal{X}^k_n$, $k=1,2,3,4$,  is  a minimum dominating  set of $\mathcal{S}^+_n$.


For the  four minimum domination sets  $\mathcal{X}^1_n$,  $\mathcal{X}^2_n$, $\mathcal{X}^3_n$,  $\mathcal{X}^4_n$ in $\mathcal{S}^+_n$, considering the vertex subset  $\mathcal{V}(\mathcal{H}_n)\subset \mathcal{V}(\mathcal{S}^+_n)$,  we have that  the types of the three extreme vertices $a_n$, $b_n$, and $c_n$ are  in succession D-D-I, D-I-D, I-D-D, and C-C-C, respectively.
We now prove that   $\mathcal{X}^1_n$,  $\mathcal{X}^2_n$, $\mathcal{X}^3_n$, and $\mathcal{X}^4_n$ are the only four possible  minimum dominating  sets of $\mathcal{S}^+_n$  by induction.  For $n=1$, $\mathcal{X}^1_1=\{ c_1\}$,  $\mathcal{X}^2_1=\{ b_1\}$, $\mathcal{X}^3_1=\{a_1\}$, and $\mathcal{X}^4_1=\{ s\}$, all of which are  the possible minimum dominating sets of $\mathcal{S}^+_1$ and are uniquely determined by the types of the three extreme vertices.

Suppose that  $\mathcal{X}^1_n$,  $\mathcal{X}^2_n$, $\mathcal{X}^3_n$, and $\mathcal{X}^4_n$ are the only minimum dominating sets of $\mathcal{S}^+_n$.  That means for any  minimum dominating set $\mathcal{X}^k_n$,  $k=1,2,3,4$,  the types of the three extreme vertices of $\mathcal{H}_n$ are determined, it is the same with the  types for extreme vertices of $\mathcal{H}_{n-2}^i$.  In other words,   the types of the three extreme vertices of $\mathcal{H}_n$ are sufficient to determine  all minimum dominating sets in $\mathcal{S}^+_n$.  For  $\mathcal{H}_{n+2}$, we next prove that  $\mathcal{X}^1_{n+2}$,  $\mathcal{X}^2_{n+2}$, $\mathcal{X}^3_{n+2}$, and $\mathcal{X}^4_{n+2}$ are the only possible  minimum dominating  sets of $\mathcal{S}^+_{n+2}$.

As shown above, with respect to any minimum dominating set $\mathcal{D}_2$  of $\mathcal{S}^+_{n+2}$,  for the three   extreme  vertices of  $\mathcal{H}_{n+2}$, either all  are of type C, or  one   is of type I, while the other two should be of type D. Given the initial types of the  three   extreme  vertices of  $\mathcal{H}_{n+2}$ as those in $\mathcal{X}^1_{n+2}$,  $\mathcal{X}^2_{n+2}$, $\mathcal{X}^3_{n+2}$, and $\mathcal{X}^4_{n+2}$,  the types of  the extreme  vertices of  its nine subgraphs $\mathcal{H}^i_{n}$, $i=1,2,\ldots,9$, are uniquely determined, belonging to D-D-I, D-I-D, I-D-D, or C-C-C,  see Fig.~\ref{fig:MiniSets}.  By induction hypothesis, given initial types the three extreme vertices $a_{n+2}$, $b_{n+2}$, and $c_{n+2}$,  the vertices  in $\mathcal{H}^i_{n}$ that are in a minimum dominating set are uniquely determined.  Consequently,  $\mathcal{X}^1_{n+2}$,  $\mathcal{X}^2_{n+2}$, $\mathcal{X}^3_{n+2}$, and $\mathcal{X}^4_{n+2}$ are the only four minimum dominating  sets of $\mathcal{S}^+_{n+2}$.
Thus the proof is completed.
\end{proof}

\begin{figure}
\centering
\includegraphics[width=0.65\textwidth]{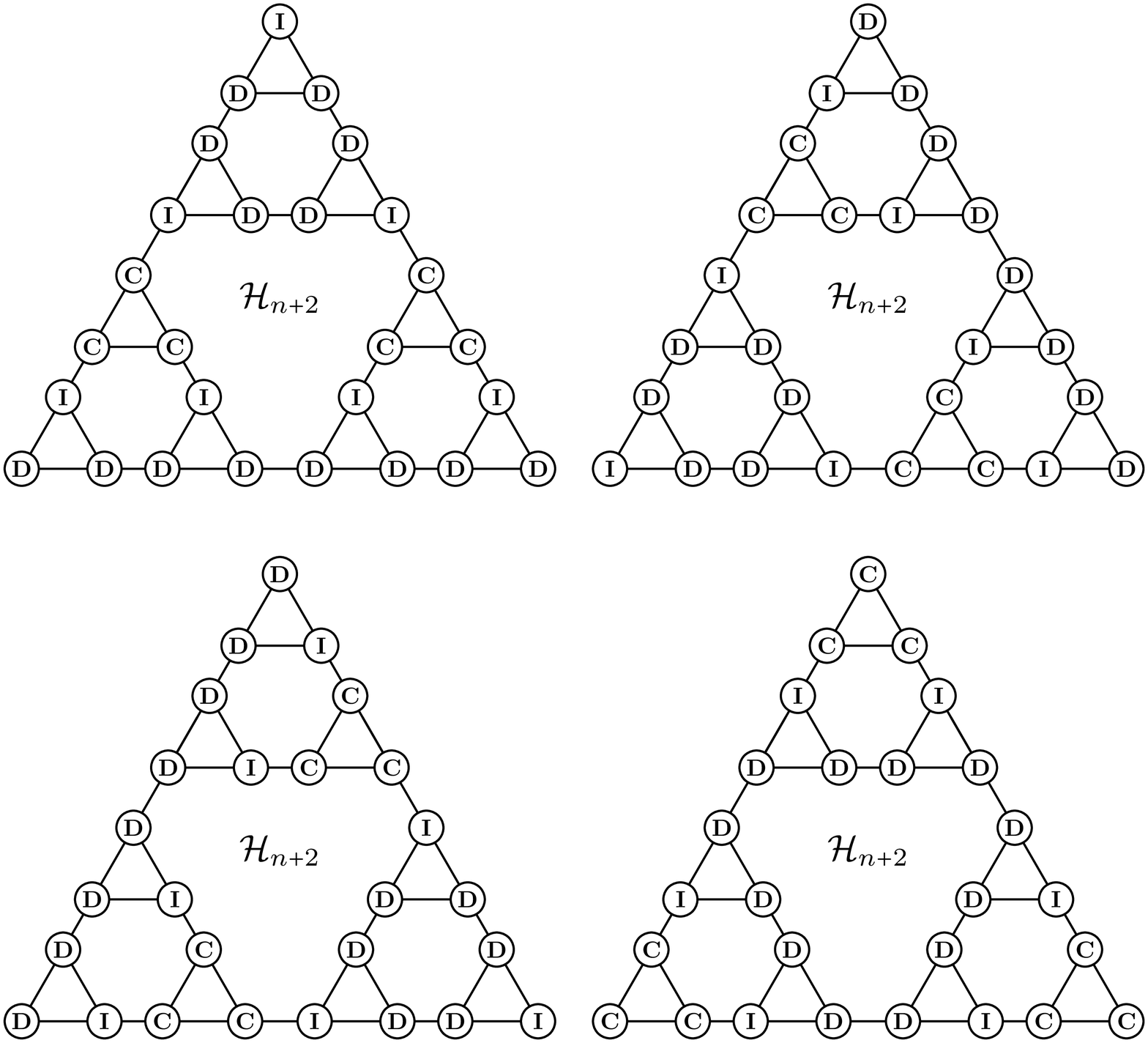}
\caption{\label{fig:MiniSets} Types of extreme vertices for  $\mathcal{H}_{n}^{i}$, $i=1,2,\ldots,9$ as  nine constituents of   $\mathcal{H}_{n+2}$,  given  the types of its three extreme vertices $a_{n+2}$,  $b_{n+2}$,  $c_{n+2}$.}
\end{figure}

Note that  1-perfect codes~\cite{KlMiPe02} and $L (2, 1)$-labelings~\cite{GrKlMo05} have been studied for Sierpi{\'n}ski graphs, which are generalization of  the Tower of  Hanoi graphs.   1-perfect codes are also called efficient dominating sets. For a graph $\mathcal{G}=(\mathcal{V},\mathcal{E})$, a subset $\mathcal{C}$ of  $\mathcal{V}$ is a  1-perfect code if and only if the neighbours of every vertex $v$ in $\mathcal{C}$, including $v$ itself,  form a partition of $\mathcal{V}$. In~\cite{KlMiPe02}, it was shown that there are exactly 3 1-perfect codes in $\mathcal{H}_n$ for odd $n$; while  in~\cite{GrKlMo05},  it was  proved that there exists an almost perfect code in  $\mathcal{H}_n$. In the first case the special vertex $s$ in not added to a perfect code, while in the second case $s$ is added to an almost perfect code. This yields  four minimum dominating sets of $\mathcal{S}^+_{n}$ as stated in Theorem~\ref{Thm32}. However, our proof for Theorem~\ref{Thm32} is short, since we only touch on the Hanoi graphs, while the two previous papers~\cite{KlMiPe02,GrKlMo05}  deal  with a broad class of Sierpi{\'n}ski graphs, subsuming the Hanoi graphs as particular cases. Thus, to some extent our work and previous ones~\cite{KlMiPe02,GrKlMo05} are  complementary to each other, at least for the Hanoi graphs $\mathcal{H}_{n}$ with odd $n$.  

Fig.~\ref{fig:31}  provides  the four different minimum dominating sets for $\mathcal{S}^+_{3}$.  Although for odd $n$, there are only four possible minimum dominating sets for $\mathcal{S}^+_{n}$, for even $n$, the number of minimum dominating sets  in  $\mathcal{S}^+_{n}$ is high. For example, numerical result shows that there are 22 different minimum dominating sets  in  $\mathcal{S}^+_{2}$. Future work should include explicit  determination of the number of minimum dominating sets  in  $\mathcal{S}^+_{n}$  for even $n$. 


\section*{Acknowledgements}

This work is supported by the National Natural Science Foundation of China under Grant No. 11275049. The authors are grateful to the anonymous reviewers for their valuable comments and suggestions, which have led to improvement of this paper.




\begin{thebibliography}{52}
\expandafter\ifx\csname natexlab\endcsname\relax\def\natexlab#1{#1}\fi
\providecommand{\bibinfo}[2]{#2}
\ifx\xfnm\relax \def\xfnm[#1]{\unskip,\space#1}\fi
\bibitem[{Lov{\'a}sz and Plummer(1986)}]{LoPl86}
\bibinfo{author}{L.~Lov{\'a}sz}, \bibinfo{author}{M.~D. Plummer},
  \bibinfo{title}{Matching Theory}, volume~\bibinfo{volume}{29} of
  \textit{\bibinfo{series}{Annals of Discrete Mathematics}},
  \bibinfo{publisher}{North Holland}, \bibinfo{address}{New York},
  \bibinfo{year}{1986}.
\bibitem[{Haynes et~al.(1998)Haynes, Hedetniemi, and Slater}]{HaHeSl98}
\bibinfo{author}{T.~W. Haynes}, \bibinfo{author}{S.~Hedetniemi},
  \bibinfo{author}{P.~Slater}, \bibinfo{title}{Fundamentals of Domination in
  Graphs}, \bibinfo{publisher}{Marcel Dekker, New York}, \bibinfo{year}{1998}.
\bibitem[{Montroll(1964)}]{Mo64}
\bibinfo{author}{E.~W. Montroll},
\newblock \bibinfo{title}{Lattice statistics},
\newblock in: \bibinfo{editor}{E.~Beckenbach} (Ed.),
  \bibinfo{booktitle}{{Applied Combinatorial Mathematics}},
  \bibinfo{publisher}{Wiley}, \bibinfo{address}{New York},
  \bibinfo{year}{1964}, pp. \bibinfo{pages}{96--143}.
\bibitem[{Vuki{\v{c}}evi{\'c}(2011)}]{Vu11}
\bibinfo{author}{D.~Vuki{\v{c}}evi{\'c}},
\newblock \bibinfo{title}{Applications of perfect matchings in chemistry},
\newblock in: \bibinfo{editor}{M.~Dehmer} (Ed.),
  \bibinfo{booktitle}{{Structural Analysis of Complex Networks}},
  \bibinfo{publisher}{Birkh{\"a}user Boston}, \bibinfo{year}{2011}, pp.
  \bibinfo{pages}{463--482}.
\bibitem[{Wu(2002)}]{Wu02}
\bibinfo{author}{J.~Wu},
\newblock \bibinfo{title}{Extended dominating-set-based routing in ad hoc
  wireless networks with unidirectional links},
\newblock \bibinfo{journal}{IEEE Trans. Parallel Distrib. Syst.}
  \bibinfo{volume}{13} (\bibinfo{year}{2002}) \bibinfo{pages}{866--881}.
\bibitem[{Shen and Li(2010)}]{ShLi10}
\bibinfo{author}{C.~Shen}, \bibinfo{author}{T.~Li},
\newblock \bibinfo{title}{Multi-document summarization via the minimum
  dominating set},
\newblock in: \bibinfo{booktitle}{Proceedings of the 23rd International
  Conference on Computational Linguistics, 2010},
  \bibinfo{organization}{Association for Computational Linguistics}, pp.
  \bibinfo{pages}{984--992}.
\bibitem[{Liu and Barab{\'a}si(2016)}]{LiBa16}
\bibinfo{author}{Y.~Y. Liu}, \bibinfo{author}{A.-L. Barab{\'a}si},
\newblock \bibinfo{title}{Control principles of complex systems},
\newblock \bibinfo{journal}{Rev. Mod. Phys.} \bibinfo{volume}{88}
  (\bibinfo{year}{2016}) \bibinfo{pages}{035006}.
\bibitem[{Liu et~al.(2011)Liu, Slotine, and Barab{\'a}si}]{LiSlBa11}
\bibinfo{author}{Y.-Y. Liu}, \bibinfo{author}{J.-J. Slotine},
  \bibinfo{author}{A.-L. Barab{\'a}si},
\newblock \bibinfo{title}{Controllability of complex networks},
\newblock \bibinfo{journal}{Nature} \bibinfo{volume}{473}
  (\bibinfo{year}{2011}) \bibinfo{pages}{167--173}.
\bibitem[{Nepusz and Vicsek(2012)}]{NeVi12}
\bibinfo{author}{T.~Nepusz}, \bibinfo{author}{T.~Vicsek},
\newblock \bibinfo{title}{Controlling edge dynamics in complex networks},
\newblock \bibinfo{journal}{Nat. Phys.} \bibinfo{volume}{8}
  (\bibinfo{year}{2012}) \bibinfo{pages}{568--573}.
\bibitem[{Nacher and Akutsu(2012)}]{NaAk12}
\bibinfo{author}{J.~C. Nacher}, \bibinfo{author}{T.~Akutsu},
\newblock \bibinfo{title}{Dominating scale-free networks with variable scaling
  exponent: heterogeneous networks are not difficult to control},
\newblock \bibinfo{journal}{New J. Phys.} \bibinfo{volume}{14}
  (\bibinfo{year}{2012}) \bibinfo{pages}{073005}.
\bibitem[{Nacher and Akutsu(2013)}]{NaAk13}
\bibinfo{author}{J.~C. Nacher}, \bibinfo{author}{T.~Akutsu},
\newblock \bibinfo{title}{Structural controllability of unidirectional
  bipartite networks},
\newblock \bibinfo{journal}{Sci. Rep.} \bibinfo{volume}{3}
  (\bibinfo{year}{2013}) \bibinfo{pages}{1647}.
\bibitem[{Klickstein et~al.(2017)Klickstein, Shirin, and Sorrentino}]{KlShSo17}
\bibinfo{author}{I.~S. Klickstein}, \bibinfo{author}{A.~Shirin},
  \bibinfo{author}{F.~Sorrentino},
\newblock \bibinfo{title}{Energy scaling of targeted optimal control of complex
  networks},
\newblock \bibinfo{journal}{Nat. Commun.} \bibinfo{volume}{8}
  (\bibinfo{year}{2017}) \bibinfo{pages}{15145}.
\bibitem[{Kuo(2004)}]{Er04}
\bibinfo{author}{E.~H. Kuo},
\newblock \bibinfo{title}{Applications of graphical condensation for
  enumerating matchings and tilings},
\newblock \bibinfo{journal}{Theor. Comput. Sci.} \bibinfo{volume}{319}
  (\bibinfo{year}{2004}) \bibinfo{pages}{29 -- 57}.
\bibitem[{Yan et~al.(2005)Yan, Yeh, and Zhang}]{YaYeZh05}
\bibinfo{author}{W.~Yan}, \bibinfo{author}{Y.-N. Yeh},
  \bibinfo{author}{F.~Zhang},
\newblock \bibinfo{title}{Graphical condensation of plane graphs: A
  combinatorial approach},
\newblock \bibinfo{journal}{Theor. Comput. Sci.} \bibinfo{volume}{349}
  (\bibinfo{year}{2005}) \bibinfo{pages}{452 -- 461}.
\bibitem[{Yan and Zhang(2005)}]{YaZh05}
\bibinfo{author}{W.~Yan}, \bibinfo{author}{F.~Zhang},
\newblock \bibinfo{title}{Graphical condensation for enumerating perfect
  matchings},
\newblock \bibinfo{journal}{J. Comb. Theory Ser. A} \bibinfo{volume}{110}
  (\bibinfo{year}{2005}) \bibinfo{pages}{113 -- 125}.
\bibitem[{Yan and Zhang(2008)}]{YaZh08}
\bibinfo{author}{W.~Yan}, \bibinfo{author}{F.~Zhang},
\newblock \bibinfo{title}{A quadratic identity for the number of perfect
  matchings of plane graphs},
\newblock \bibinfo{journal}{Theor. Comput. Sci.} \bibinfo{volume}{409}
  (\bibinfo{year}{2008}) \bibinfo{pages}{405--410}.
\bibitem[{Teufl and Wagner(2009)}]{TeSt09}
\bibinfo{author}{E.~Teufl}, \bibinfo{author}{S.~Wagner},
\newblock \bibinfo{title}{Exact and asymptotic enumeration of perfect matchings
  in self-similar graphs},
\newblock \bibinfo{journal}{Discrete Math.} \bibinfo{volume}{309}
  (\bibinfo{year}{2009}) \bibinfo{pages}{6612 -- 6625}.
\bibitem[{Chebolu et~al.(2010)Chebolu, Frieze, and Melsted}]{ChFrMe10}
\bibinfo{author}{P.~Chebolu}, \bibinfo{author}{A.~Frieze},
  \bibinfo{author}{P.~Melsted},
\newblock \bibinfo{title}{{Finding a maximum matching in a sparse random graph
  in $O (n)$ expected time}},
\newblock \bibinfo{journal}{J. ACM} \bibinfo{volume}{57} (\bibinfo{year}{2010})
  \bibinfo{pages}{24}.
\bibitem[{Yuster(2013)}]{Yu13}
\bibinfo{author}{R.~Yuster},
\newblock \bibinfo{title}{Maximum matching in regular and almost regular
  graphs},
\newblock \bibinfo{journal}{Algorithmica} \bibinfo{volume}{66}
  (\bibinfo{year}{2013}) \bibinfo{pages}{87--92}.
\bibitem[{Zhang and Wu(2015)}]{ZhWu15}
\bibinfo{author}{Z.~Zhang}, \bibinfo{author}{B.~Wu},
\newblock \bibinfo{title}{Pfaffian orientations and perfect matchings of
  scale-free networks},
\newblock \bibinfo{journal}{Theoret. Comput. Sci.} \bibinfo{volume}{570}
  (\bibinfo{year}{2015}) \bibinfo{pages}{55--69}.
\bibitem[{Meghanathan(2016)}]{Me16}
\bibinfo{author}{N.~Meghanathan},
\newblock \bibinfo{title}{Maximal assortative matching and maximal dissortative
  matching for complex network graphs},
\newblock \bibinfo{journal}{Comput. J.} \bibinfo{volume}{59}
  (\bibinfo{year}{2016}) \bibinfo{pages}{667--684}.
\bibitem[{Li and Zhang(2017)}]{LiZh17}
\bibinfo{author}{H.~Li}, \bibinfo{author}{Z.~Zhang},
\newblock \bibinfo{title}{{Maximum matchings in scale-free networks with
  identical degree distribution}},
\newblock \bibinfo{journal}{Theoret. Comput. Sci.} \bibinfo{volume}{675}
  (\bibinfo{year}{2017}) \bibinfo{pages}{64--81}.
\bibitem[{Fomin et~al.(2008)Fomin, Grandoni, Pyatkin, and
  Stepanov}]{FoGrPySt08}
\bibinfo{author}{F.~V. Fomin}, \bibinfo{author}{F.~Grandoni},
  \bibinfo{author}{A.~V. Pyatkin}, \bibinfo{author}{A.~A. Stepanov},
\newblock \bibinfo{title}{Combinatorial bounds via measure and conquer:
  Bounding minimal dominating sets and applications},
\newblock \bibinfo{journal}{ACM Tran. Algorithms} \bibinfo{volume}{5}
  (\bibinfo{year}{2008}) \bibinfo{pages}{9}.
\bibitem[{Hedar and Ismail(2012)}]{HeIs12}
\bibinfo{author}{A.-R. Hedar}, \bibinfo{author}{R.~Ismail},
\newblock \bibinfo{title}{Simulated annealing with stochastic local search for
  minimum dominating set problem},
\newblock \bibinfo{journal}{Int. J. Mach. Learn. Cybernet.} \bibinfo{volume}{3}
  (\bibinfo{year}{2012}) \bibinfo{pages}{97--109}.
\bibitem[{da~Fonseca et~al.(2014)da~Fonseca, de~Figueiredo, de~S{\'a}, and
  Machado}]{dadedeMa14}
\bibinfo{author}{G.~D. da~Fonseca}, \bibinfo{author}{C.~M. de~Figueiredo},
  \bibinfo{author}{V.~G.~P. de~S{\'a}}, \bibinfo{author}{R.~C. Machado},
\newblock \bibinfo{title}{Efficient sub-5 approximations for minimum dominating
  sets in unit disk graphs},
\newblock \bibinfo{journal}{Theoret. Comput. Sci.} \bibinfo{volume}{540}
  (\bibinfo{year}{2014}) \bibinfo{pages}{70--81}.
\bibitem[{Gast et~al.(2015)Gast, Hauptmann, and Karpinski}]{GaHaK15}
\bibinfo{author}{M.~Gast}, \bibinfo{author}{M.~Hauptmann},
  \bibinfo{author}{M.~Karpinski},
\newblock \bibinfo{title}{Inapproximability of dominating set on power law
  graphs},
\newblock \bibinfo{journal}{Theoret. Comput. Sci.} \bibinfo{volume}{562}
  (\bibinfo{year}{2015}) \bibinfo{pages}{436--452}.
\bibitem[{Couturier et~al.(2015)Couturier, Letourneur, and Liedloff}]{CoLeLi15}
\bibinfo{author}{J.-F. Couturier}, \bibinfo{author}{R.~Letourneur},
  \bibinfo{author}{M.~Liedloff},
\newblock \bibinfo{title}{On the number of minimal dominating sets on some
  graph classes},
\newblock \bibinfo{journal}{Theoret. Comput. Sci.} \bibinfo{volume}{562}
  (\bibinfo{year}{2015}) \bibinfo{pages}{634--642}.
\bibitem[{Shan et~al.(2017)Shan, Li, and Zhang}]{ShLiZh17}
\bibinfo{author}{L.~Shan}, \bibinfo{author}{H.~Li}, \bibinfo{author}{Z.~Zhang},
\newblock \bibinfo{title}{{Domination number and minimum dominating sets in
  pseudofractal scale-free web and Sierpi\'nski graph}},
\newblock \bibinfo{journal}{Theoret. Comput. Sci.} \bibinfo{volume}{677}
  (\bibinfo{year}{2017}) \bibinfo{pages}{12--30}.
\bibitem[{Valiant(1979{\natexlab{a}})}]{Va79TCS}
\bibinfo{author}{L.~Valiant},
\newblock \bibinfo{title}{The complexity of computing the permanent},
\newblock \bibinfo{journal}{Theor. Comput. Sci.} \bibinfo{volume}{8}
  (\bibinfo{year}{1979}{\natexlab{a}}) \bibinfo{pages}{189--201}.
\bibitem[{Valiant(1979{\natexlab{b}})}]{Va79SiamJComput}
\bibinfo{author}{L.~Valiant},
\newblock \bibinfo{title}{The complexity of enumeration and reliability
  problems},
\newblock \bibinfo{journal}{SIAM J. Comput.} \bibinfo{volume}{8}
  (\bibinfo{year}{1979}{\natexlab{b}}) \bibinfo{pages}{410--421}.
\bibitem[{Assadi et~al.(2017)Assadi, Khanna, and Li}]{AsKhLi17}
\bibinfo{author}{S.~Assadi}, \bibinfo{author}{S.~Khanna},
  \bibinfo{author}{Y.~Li},
\newblock \bibinfo{title}{On estimating maximum matching size in graph
  streams},
\newblock in: \bibinfo{booktitle}{{Proceedings of the 28th annual ACM-SIAM
  symposium on discrete algorithms, 2017}}, \bibinfo{organization}{SIAM}, pp.
  \bibinfo{pages}{1723--1742}.
\bibitem[{Golovach et~al.(2016)Golovach, Heggernes, and Kratsch}]{GoHeKr16}
\bibinfo{author}{P.~A. Golovach}, \bibinfo{author}{P.~Heggernes},
  \bibinfo{author}{D.~Kratsch},
\newblock \bibinfo{title}{Enumerating minimal connected dominating sets in
  graphs of bounded chordality},
\newblock \bibinfo{journal}{Theoret. Comput. Sci.} \bibinfo{volume}{630}
  (\bibinfo{year}{2016}) \bibinfo{pages}{63--75}.
\bibitem[{Andrade~Jr et~al.(2005)Andrade~Jr, Herrmann, Andrade, and
  Da~Silva}]{AnHeAnDa05}
\bibinfo{author}{J.~S. Andrade~Jr}, \bibinfo{author}{H.~J. Herrmann},
  \bibinfo{author}{R.~F. Andrade}, \bibinfo{author}{L.~R. Da~Silva},
\newblock \bibinfo{title}{Apollonian networks: Simultaneously scale-free, small
  world, euclidean, space filling, and with matching graphs},
\newblock \bibinfo{journal}{Phys. Rev. Lett.} \bibinfo{volume}{94}
  (\bibinfo{year}{2005}) \bibinfo{pages}{018702}.
\bibitem[{Klav{\v{z}}ar and Mohar(2005)}]{KlMo05}
\bibinfo{author}{S.~Klav{\v{z}}ar}, \bibinfo{author}{B.~Mohar},
\newblock \bibinfo{title}{{Crossing numbers of Sierpi{\'n}ski-like graphs}},
\newblock \bibinfo{journal}{J. Graph Theory} \bibinfo{volume}{50}
  (\bibinfo{year}{2005}) \bibinfo{pages}{186--198}.
\bibitem[{Barab{\'a}si and Albert(1999)}]{BaAl99}
\bibinfo{author}{A.~Barab{\'a}si}, \bibinfo{author}{R.~Albert},
\newblock \bibinfo{title}{Emergence of scaling in random networks},
\newblock \bibinfo{journal}{Science} \bibinfo{volume}{286}
  (\bibinfo{year}{1999}) \bibinfo{pages}{509--512}.
\bibitem[{Watts and Strogatz(1998)}]{WaSt98}
\bibinfo{author}{D.~Watts}, \bibinfo{author}{S.~Strogatz},
\newblock \bibinfo{title}{Collective dynamics of `small-world' networks},
\newblock \bibinfo{journal}{Nature} \bibinfo{volume}{393}
  (\bibinfo{year}{1998}) \bibinfo{pages}{440--442}.
\bibitem[{Newman(2003)}]{Ne03}
\bibinfo{author}{M.~E.~J. Newman},
\newblock \bibinfo{title}{The structure and function of complex networks},
\newblock \bibinfo{journal}{SIAM Rev.} \bibinfo{volume}{45}
  (\bibinfo{year}{2003}) \bibinfo{pages}{167--256}.
\bibitem[{Chang et~al.(2007)Chang, Chen, and Yang}]{ChChY07}
\bibinfo{author}{S.-C. Chang}, \bibinfo{author}{L.-C. Chen},
  \bibinfo{author}{W.-S. Yang},
\newblock \bibinfo{title}{Spanning trees on the {S}ierpi{\'n}ski gasket},
\newblock \bibinfo{journal}{J. Stat. Phys.} \bibinfo{volume}{126}
  (\bibinfo{year}{2007}) \bibinfo{pages}{649--667}.
\bibitem[{Zhang et~al.(2008)Zhang, Chen, Zhou, Fang, Guan, and
  Zou}]{ZhChZhFaGuZo08}
\bibinfo{author}{Z.~Zhang}, \bibinfo{author}{L.~Chen},
  \bibinfo{author}{S.~Zhou}, \bibinfo{author}{L.~Fang},
  \bibinfo{author}{J.~Guan}, \bibinfo{author}{T.~Zou},
\newblock \bibinfo{title}{{Analytical solution of average path length for
  Apollonian networks}},
\newblock \bibinfo{journal}{Phys. Rev. E} \bibinfo{volume}{77}
  (\bibinfo{year}{2008}) \bibinfo{pages}{017102}.
\bibitem[{Zhang et~al.(2013)Zhang, Sun, and Xu}]{ZhSuXu13}
\bibinfo{author}{J.~Zhang}, \bibinfo{author}{W.~Sun}, \bibinfo{author}{G.~Xu},
\newblock \bibinfo{title}{{Enumeration of spanning trees on Apollonian
  networks}},
\newblock \bibinfo{journal}{J. Stat. Mech.} \bibinfo{volume}{2013}
  (\bibinfo{year}{2013}) \bibinfo{pages}{P09015}.
\bibitem[{Zhang et~al.(2014)Zhang, Wu, and Comellas}]{ZhWuCo14}
\bibinfo{author}{Z.~Zhang}, \bibinfo{author}{B.~Wu},
  \bibinfo{author}{F.~Comellas},
\newblock \bibinfo{title}{The number of spanning trees in {Apollonian}
  networks},
\newblock \bibinfo{journal}{Discrete Appl. Math.} \bibinfo{volume}{169}
  (\bibinfo{year}{2014}) \bibinfo{pages}{206 -- 213}.
\bibitem[{Zhang and Mahmoud(2016)}]{ZhMa16}
\bibinfo{author}{P.~Zhang}, \bibinfo{author}{H.~Mahmoud},
\newblock \bibinfo{title}{{The degree profile and weight in Apollonian networks
  and $k$-trees}},
\newblock \bibinfo{journal}{Adv. Appl. Prob.} \bibinfo{volume}{48}
  (\bibinfo{year}{2016}) \bibinfo{pages}{163--175}.
\bibitem[{Hinz et~al.(2013)Hinz, Klav{\v{z}}ar, Milutinovi{\'c}, and
  Petr}]{HiKlMiPeSt13}
\bibinfo{author}{A.~M. Hinz}, \bibinfo{author}{S.~Klav{\v{z}}ar},
  \bibinfo{author}{U.~Milutinovi{\'c}}, \bibinfo{author}{C.~Petr},
  \bibinfo{title}{The Tower of {H}anoi-- Myths and Maths},
  \bibinfo{publisher}{Springer}, \bibinfo{year}{2013}.
\bibitem[{Savage(1997)}]{Sa97}
\bibinfo{author}{C.~Savage},
\newblock \bibinfo{title}{{A survey of combinatorial Gray codes}},
\newblock \bibinfo{journal}{SIAM Rev.} \bibinfo{volume}{39}
  (\bibinfo{year}{1997}) \bibinfo{pages}{605--629}.
\bibitem[{Hinz et~al.(2017)Hinz, Klav{\v{z}}ar, and Zemlji{\v{c}}}]{HiklSa17}
\bibinfo{author}{A.~M. Hinz}, \bibinfo{author}{S.~Klav{\v{z}}ar},
  \bibinfo{author}{S.~S. Zemlji{\v{c}}},
\newblock \bibinfo{title}{A survey and classification of {S}ierpi{\'n}ski-type
  graphs},
\newblock \bibinfo{journal}{Discrete Appl. Math.} \bibinfo{volume}{217}
  (\bibinfo{year}{2017}) \bibinfo{pages}{565--600}.
\bibitem[{Chen et~al.(2015)Chen, Wu, Huang, and Deng}]{ChWuHuDe15}
\bibinfo{author}{H.~Chen}, \bibinfo{author}{R.~Wu}, \bibinfo{author}{G.~Huang},
  \bibinfo{author}{H.~Deng},
\newblock \bibinfo{title}{{Dimer--monomer model on the Towers of Hanoi
  graphs}},
\newblock \bibinfo{journal}{Int. J. Mod. Phys. B} \bibinfo{volume}{29}
  (\bibinfo{year}{2015}) \bibinfo{pages}{1550173}.
\bibitem[{Chang and Chen(2008)}]{ChCh08}
\bibinfo{author}{S.-C. Chang}, \bibinfo{author}{L.-C. Chen},
\newblock \bibinfo{title}{Dimer coverings on the {S}ierpi{\'n}ski gasket},
\newblock \bibinfo{journal}{J. Stat. Phys.} \bibinfo{volume}{131}
  (\bibinfo{year}{2008}) \bibinfo{pages}{631--650}.
\bibitem[{Teufl and Wagner(2007)}]{TeWa07}
\bibinfo{author}{E.~Teufl}, \bibinfo{author}{S.~Wagner},
\newblock \bibinfo{title}{Enumeration problems for classes of self-similar
  graphs},
\newblock \bibinfo{journal}{J. Comb. Theory Ser. A} \bibinfo{volume}{114}
  (\bibinfo{year}{2007}) \bibinfo{pages}{1254--1277}.
\bibitem[{Li and Nelson(1998)}]{LiNe98}
\bibinfo{author}{C.-K. Li}, \bibinfo{author}{I.~Nelson},
\newblock \bibinfo{title}{{Perfect codes on the Towers of Hanoi graph}},
\newblock \bibinfo{journal}{Bull. Austral. Math. Soc.} \bibinfo{volume}{57}
  (\bibinfo{year}{1998}) \bibinfo{pages}{367--376}.
\bibitem[{Cull and Nelson(1999)}]{CuNe99}
\bibinfo{author}{P.~Cull}, \bibinfo{author}{I.~Nelson},
\newblock \bibinfo{title}{{Error-correcting codes on the Towers of Hanoi
  graphs}},
\newblock \bibinfo{journal}{Discrete Math.} \bibinfo{volume}{208}
  (\bibinfo{year}{1999}) \bibinfo{pages}{157--175}.
\bibitem[{Klav{\v{z}}ar et~al.(2002)Klav{\v{z}}ar, Milutinovi{\'c}, and
  Petr}]{KlMiPe02}
\bibinfo{author}{S.~Klav{\v{z}}ar}, \bibinfo{author}{U.~Milutinovi{\'c}},
  \bibinfo{author}{C.~Petr},
\newblock \bibinfo{title}{{1-perfect codes in Sierpi{\'n}ski graphs}},
\newblock \bibinfo{journal}{Bull. Austral. Math. Soc.} \bibinfo{volume}{66}
  (\bibinfo{year}{2002}) \bibinfo{pages}{369--384}.
\bibitem[{Gravier et~al.(2005)Gravier, Klav{\v{z}}ar, and Mollard}]{GrKlMo05}
\bibinfo{author}{S.~Gravier}, \bibinfo{author}{S.~Klav{\v{z}}ar},
  \bibinfo{author}{M.~Mollard},
\newblock \bibinfo{title}{{Codes and $L (2, 1)$-labelings in Sierpi{\'n}ski
  graphs}},
\newblock \bibinfo{journal}{Taiwanese J. Math.}  (\bibinfo{year}{2005})
  \bibinfo{pages}{671--681}.

\end{thebibliography}

\end{document}